\documentclass[pre,notitlepage,twocolumn]{revtex4-1}

\usepackage{amsmath}
\usepackage{amsfonts}
\usepackage{tikz}
\usepackage{graphicx}
\usepackage{subfigure}
\usepackage{pdfsync}
\usepackage{bm}
\usepackage{color}
\usepackage[noend]{algpseudocode}
\usepackage[justification=raggedright,singlelinecheck=false]{caption}
\usepackage{enumerate}
\DeclareCaptionType{algorithm}
\usepackage{algorithmicx}
\algdef{SE}[DOWHILE]{Do}{doWhile}{\algorithmicdo}[1]{\algorithmicwhile\ #1}%
\usepackage{verbatim}
\usepackage{amsthm}
\usepackage{url}
\usepackage{caption}

\newcommand{\ket}[1]{\ensuremath {|\: #1 \: \rangle}}
\newcommand{\bra}[1]{\ensuremath{\langle \: #1 \:|}}
\newcommand{\braket}[2]{\ensuremath{\langle \: #1 \: | \: #2 \: \rangle}}

\newcommand{\eref}[1]{(\ref{#1})}
\newcommand{\fref}[1]{figure \ref{#1}}

\newcommand{\sref}[1]{section \ref{#1}}

\newcommand{\llrr}[1]{\ensuremath{\left( #1\right)}}
\newcommand{\llrrq}[1]{\ensuremath{\left[ #1\right]}}

\newtheorem{theorem}{Theorem}
\newtheorem{fact}{Fact}
\newtheorem{corollary}{Corollary}
\begin{document}
\title{Improved scaling of Time-Evolving Block-Decimation algorithm through Reduced-Rank Randomized Singular Value Decomposition}
\author{D. Tamascelli$^{1,2}$, R. Rosenbach$^{2}$, and M.~B. Plenio$^{2}$}
\affiliation{$^{1}$ Dipartimento di Fisica, Universit{\`a} degli Studi di Milano, Via Celoria 16, 20133 Milano-Italy\\
$^{2}$ Institut f\"ur Theoretische Physik \& IQST, Albert-Einstein-Allee 11, Universit\"at Ulm, Germany}
%dario.tamascelli@unimi.it
%
\begin{abstract}

When the amount of entanglement in a quantum system is limited, the relevant
dynamics of the system is restricted to a very small part of the state space. When
restricted to this subspace the description of the system becomes efficient in the
system size. A class of algorithms, exemplified by the Time-Evolving Block-Decimation
(TEBD) algorithm, make use of this observation by selecting the relevant subspace
through a decimation technique relying on the Singular Value Decomposition (SVD).
%This allows the efficient simulation of one-dimensional quantum many-body systems with nearest-neighbor interactions.
In these algorithms, the complexity of each time-evolution step is dominated by the SVD.
Here we show that, by applying a randomized version of the SVD routine (RRSVD), the
power law governing the computational complexity of TEBD is lowered by one degree,
resulting in a considerable speed-up. We exemplify the potential gains in efficiency
at the hand of some real world examples to which TEBD can be successfully applied to
and demonstrate that for those system RRSVD delivers results as accurate as state-of-the-art
deterministic SVD routines.
\end{abstract}
\maketitle
%\tableofcontents
%
\section{Introduction}
\label{Introduction}

The description of the physics of quantum many-body systems suffers from the "curse of dimensionality",
that is, the size of the parameter set that is required to achieve an exact description of the physical
state of a quantum many-body system grows exponentially in the number of its subsystems. Therefore, the
simulation of quantum-many-body systems by classical means appears to require an exponential amount of
computational resources which, in turn, would impose severe limitations on the size of the quantum
many-body systems that are amenable to classical simulation.

On the other hand, exactness of description may be traded for an approximate representations of the
state and dynamics of a quantum many-body system for as long as the quality of the approximation can
be controlled and increased at will and whenever such an approximate treatment results in a polynomial
scaling of the resources with the system size. Clearly, this will not be possible for all system dynamics
as it would imply the classical simulability of quantum computers which is generally not believed to
be the case.

One specific setting of considerable practical importance that allows for efficient approximate
description of quantum many-body systems concerns systems whose entanglement content is limited.
Indeed, in pure states that factor, at all times, into a product of
states; each involving a number of qubits that is bounded from above for all times by a constant
\cite{jozsa03} can be simulated efficiently on a classical device. Going beyond this, it is well-known
that the states of 1-D quantum systems often obey an area law \cite{Audenaert2002,Plenio2005,eisert08,horodecki13,horodecki15}
which represents a severe limitation of their entanglement content. This can be made use of, as
the state of such slightly entangled 1-D quantum many-body system in a pure state can be described
efficiently by means of the Density Matrix Renormalization Group (DMRG) \cite{white92}. It was
noticed early on \cite{ref:rommer1997} that DMRG amounts to the approximation of the state of the
system by a matrix product state \cite{ref:perezgarcia2007} which can be constructed systematically
in terms of a consecutive Schmidt decomposition.

This approach can be extended to the dynamics of one-dimensional quantum many-body systems. The
time-dependent DMRG (t-DMRG) \cite{white04}, the time-dependent Matrix Product States (t-MPS)
\cite{garcia06}, and the Time-Evolving Block-Decimation (TEBD) \cite{ref:vidal2003,ref:vidal2004}
are all algorithms based on the idea of evolving an MPS \cite{ref:perezgarcia2007} in time. In
settings in which only a restricted amount of entanglement is present in the system, these methods
are very efficient and the cost of simulation scales merely polynomially in the system size. A very
clear presentation of these three algorithms as well as a discussion of the main differences
between them can be found in \cite{ref:schollwoeck2011}.

All these algorithms rely in an essential way on the singular value decomposition (SVD) which, due
to its computational complexity, represents a bottleneck in their implementations. In this work we
show how to address this issue within the TEBD framework by replacing the SVD which is used to restrict
the Hilbert space of quantum states relevant to the system dynamics, with the Reduced-Rank Randomized
SVD (RRSVD) proposed in \cite{halko11}. We show that when the Schmidt coefficients decay exponentially,
the polynomial complexity of the TEBD decimation step can be reduced from $\mathcal{O} (n^3)$ to
$\mathcal{O}(n^2)$ with $n$ indicating the linear dimension of the square matrix to decompose. This 
results in a considerable speed-up for the TEBD algorithm for real world problems, enabling the access 
of larger parameter spaces, faster simulation times, and thus opening up new regimes that were previously 
inaccessible.

The paper is organized as follows: Section \ref{sec:tebd} provides a brief description of the TEBD algorithm for
pure and mixed state dynamics in order to make the manuscript self-contained for the non-specialist
and to identify and highlight the crucial step in which the RRSVD routine can be applied for considerable
benefit. Section \ref{sec:tedopa} introduces a specific and challenging application of TEBD, the TEDOPA scheme, which
maps an open quantum system on a one-dimensional configuration, allowing for the efficient simulation of
its dynamics. Section \ref{sec:RRSVD} the proceeds with a description of the salient features of the RRSVD
algorithm and a discussion of the speed-up that it provides over the standard SVD. Benchmarking in the
TEBD context with applications to TEDOPA along with stability analysis are presented in section \ref{sec:rrsvdtedopa}.
The last section is devoted to conclusions and outlook.

\section{An algorithm for one-dimensional quantum systems}
\label{sec:tebd}
The time evolving block decimation (TEBD) is an algorithm that generates
efficiently an approximation to the time evolution of a one-dimensional
system subject to a nearest-neighbor Hamiltonian. Under the condition
that the amount of entanglement in the system is bounded a high fidelity
approximation requires polynomially scaling computational resources. TEBD
does so by dynamically restricting the exponentially large Hilbert space
to its most relevant subspace whose size is scaling polynomially in the
system size, thus rendering the computation feasible \cite{ref:schollwoeck2011, ref:vidal2004}.

TEBD is essentially a combination of an MPS description for a one-dimensional
quantum system and an algorithm that applies two-site gates that are necessary
to implement a Suzuki-Trotter time evolution. Together with MPS operations such
as the application of measurements this yields a powerful simulation framework
\cite{ref:perezgarcia2007}.

While originally formulated for pure states, an extension to mixed states is
possible by introducing a matrix product operator (MPO) to describe the density
matrix, in complete analogy to an MPS describing a state \cite{ref:zwolak2004}.
The simulation procedure remains unchanged, except for details such as a squaring
of the local dimension on each site, the manner in which two-site gates are
built, the procedures to achieve normalisation as well as the implementation
of measurements \cite{ref:zwolak2004, ref:schollwoeck2011}. While standard
MPO formulations cannot ensure positivity of the state efficiently, recent
reformulations can account for this feature too \cite{Montangero14}.
Important for the present work are implications of these modifications on the
two-site update - and while the numerical recipe does not change, its scaling {\em does}.

\subsection{Introduction to MPS}

The remarkable usefulness and broad range of applications of the MPS
description for quantum states has been widely recognized early on in
the development of DMRG algorithms \cite{ref:rommer1997}. To better 
understand the full extent of the presented work, we highlight the 
relevant key features of MPS, referring to references~\cite{ref:perezgarcia2007,ref:schollwoeck2011}
for a full account.

Let us introduce an MPS for a pure state of $N$ sites. For simplicity
we assume that each site has the same number of dimensions~$d$, the
extension to varying dimension is straight forward. The MPS then
relates the expansion coefficients $c_{i_1 i_2 \ldots i_N}$ in the Fock
basis to a set of $N\cdot d$~matrices $\Gamma$ and $N-1$ matrices
$\lambda$
\begin{align}
	\ket{\psi} &=	\sum_{i_1, i_2, \ldots i_N}
			c_{i_1 i_2 \ldots i_N}
			\ket{i_1 i_2 \cdots i_N}
			\label{eq:defstate} \\
		    &=	\sum_{i_1, i_2, \ldots i_N}
		   	\Gamma^{\llrrq{1}i_1} \cdot \lambda^{\llrrq{1}}
			\cdot \Gamma^{\llrrq{2}i_2} \cdot \ldots \cdot \nonumber \\
			& \hspace{2.cm}\ldots \lambda^{\llrrq{N-1}} \Gamma^{\llrrq{N}i_N}
			\ket{i_1 i_2 \cdots i_N}.
			\label{eq:defmps}
\end{align}

Each of the $N$ sites is assigned a set of $d$ matrices~$\Gamma$ which
have dimension $\chi_l\times\chi_r$. The index~$k$ in square brackets
denotes the corresponding site and the $i_k$ the corresponding physical
state. The diagonal $\chi_b\times\chi_b$ matrices $\lambda$ are assigned to
the bond~$k$ between sites~$k$ and $k+1$. The structure of the MPS is
such that the matrices $\lambda$ contain the Schmidt values for a
bipartition at this bond.  The matrices~$\Gamma$ and $\lambda$ are
related to the coefficients~$c$ by
\begin{equation}
	c_{i_1 i_2 \ldots i_N} = \Gamma^{\llrrq{1}i_1} \cdot
	\lambda^{\llrrq{1}} \cdot \Gamma^{\llrrq{2}i_2} \cdot
	\ldots \cdot \lambda^{\llrrq{N-1}} \cdot
	\Gamma^{\llrrq{N}i_N}
	\label{eq:relcoefmat}
\end{equation}
with matrix dimensions $\chi_l$, $\chi_r$, and $\chi_b$ for all $\Gamma$
and $\lambda$ such that building the product yields a number. The main
reason for employing an MPS description is the reduction from $d^N$
coefficients~$c$ to only $\mathcal{O}\llrr{d N \chi^2}$ when the matrices
$\Gamma$ and $\lambda$ are at most of size $\chi\times\chi$. This description
is efficient, provided that the matrix size $\chi$ (also known as the
{\em bond dimension}) is restricted - which it is, if the amount of
entanglement in the system is bounded. Further the MPS structure entails that
local gates applied to the whole state only change the matrices of the
sites they act on - thus updates are inexpensive, as opposed to a full
state vector description.

\subsection{Two-site gates}

The crucial ingredient in this simulation scheme is the SVD. It is the
solution to the question of how to apply a two-site gate to an MPS or MPO.

The fact that a gate $G$ acting on the two sites $k$ and $k+1$ only
changes the matrices local to these sites can easily be seen from
\begin{align}
	G_{k,k+1} \ket{\psi}
	= &\sum_{i_1 \ldots i_N}
	   \Gamma^{\llrrq{1}i_1} \cdot \lambda^{\llrrq{1}} \cdot \ldots
	   \cdot \Gamma^{\llrrq{N}i_N} \\
	  & \ket{i_1 \ldots i_{k-1}}  G_{k,k+1} \ket{i_k i_{k+1}} \ket{i_{k+2} \cdots i_N}  \nonumber \\
	=& \sum_{i_1 \ldots i_N} \sum_{j_k, j_{k+1}}
	   \Gamma^{\llrrq{1}i_1} \cdot \lambda^{\llrrq{1}} \cdot \ldots \cdot \Gamma^{\llrrq{N}i_N}\nonumber \\
	  &  \bra{i_k i_{k+1}}
	   G_{k,k+1} \ket{j_k j_{k+1}} \ket{i_1 \ldots i_N}.
	\label{eq:deftsgate}
\end{align}
Above we first introduced the completeness relation for $j_k$ and
$j_{k+1}$, followed by switching indices $i_k$ with $j_k$ and $i_{k+1}$
with $j_{k+1}$. Identifying all terms related to $j_k$ and $j_{k+1}$
now defines a tensor $\Theta$ of fourth order
\begin{equation}
	\Theta\llrr{i_k, i_{k+1}, \alpha, \beta} =
	\sum_{\gamma} \lambda_{\alpha}^{\llrrq{k-1}}
	\Gamma_{\alpha,\gamma}^{\llrrq{k}i_k} \cdot
	\lambda_{\gamma}^{\llrrq{k}} \cdot
	\Gamma_{\gamma,\beta}^{\llrrq{k+1}i_{k+1}} \cdot
	\lambda_{\beta}^{\llrrq{k+1}}
	\label{eq:deftheta}
\end{equation}
which can be built at a numerical cost of $\mathcal{O}\llrr{d_k \cdot
d_{k+1} \cdot \chi^3}$ basic operations. This tensor needs to be
updated when applying the gate~$G$. The update rule from
Eq.~\eqref{eq:deftsgate} yields the relation
\begin{align}
	\tilde{\Theta}\llrr{i_k, i_{k+1}, \alpha, \beta} =&
	\sum_{j_k, j_{k+1}}
	\bra{i_k i_{k+1}} G_{k,k+1} \ket{j_k j_{k+1}} \cdot \nonumber \\
	&\cdot \Theta\llrr{j_k, j_{k+1}, \alpha, \beta}.
	\label{eq:defupdatetheta}
\end{align}

This sum is performed for all $\alpha$ and $\beta$ - which in general
run from $1$ to $\chi$. Thus there are $d_k\times d_{k+1}$ products,
which gives the number of basic operations for the update of $\Theta$
as $\mathcal{O}\llrr{d_l^2 \cdot d_r^2 \cdot \chi^2}$
\cite{ref:vidal2003}. This formula however only enables the update of
the complete matrix $\Theta$ and not of the single entities
$\Gamma^{\llrrq{k}}$, $\lambda^{\llrrq{k}}$ and $\Gamma^{\llrrq{k+1}}$.
They still have to be extracted from these updated products. To do
this, $\Theta$ is written in a blocked index form
$\Theta_{\llrr{d_k\chi},\llrr{d_{k+1}\chi}}$ which then is singular
value-decomposed \cite{ref:vidal2003}. The general singular value
decomposition (SVD) scales as $\mathcal{O}\llrr{m\cdot n^2}$ for a
$m\times n$-matrix, thus resulting in an overall computational cost of
$\mathcal{O}\llrr{d_l \cdot d_r^2 \cdot \chi^3}$. This makes the SVD
the real bottleneck in the simulation, consuming by far the most resources,
and therefore the first target for improvements.

Here we stick to the (unofficial) standard notation for MPS parameters
in the context of the TEBD algorithm and denote the diagonal matrices
as well as their entries by $\lambda$. During our discussion of the
singular value decomposition though we switch to the respective
notational standards where the $i$'th singular value will be denoted by
$\sigma_i$ - which however is the same quantity as the $\lambda_i$ in
this section.

\subsection{Error analysis}

During a TEBD simulation, the two main error sources are the Trotter
and the truncation error. Other small error sources, often depending on
implementational or algorithmical choices, are neglected in the
following analysis.

TEBD relies heavily on the nearest-neighbor structure of the underlying
Hamiltonian to implement time evolution. Traditionally this is done by
standard Suzuki-Trotter decompositions \cite{ref:suzuki1990} where the
total Hamiltonian is split into two terms $H=F+G$ with each $F$ and $G$
are the sum over all even and odd (respectively) terms of the
Hamiltonian. This way all terms within $F$ ($G$) commute with each
other, incurring no error while applying an operator of the form
$\text{exp}\llrr{\alpha F}$ ($\alpha \in \mathbb{C}$). The most
straight forward and illustrative example is the standard $3$rd-order
expansion

\begin{align}\label{eq:defst3rd}
	&\text{exp}\llrr{i H \delta t} =
	\text{exp}\llrr{i \llrr{F+G} \delta t} =  \\
	&=\text{exp}\llrr{i \frac{1}{2} F \delta t} \cdot
	\text{exp}\llrr{i G \delta t} \cdot
	\text{exp}\llrr{i \frac{1}{2} F \delta t} +
	\mathcal{O}{\llrr{\delta t}^3}. \nonumber
\end{align}

This leads to three sweeps ($\frac{1}{2}F$, $G$, $\frac{1}{2}F$) of
non-overlapping gate applications, comprising one time step.  Various
higher-order schemes for such decompositions exists, differing in the
number of sweeps and the order of the resulting error \cite{ref:suzuki2005}.
However, these kind of schemes may introduce non-orthogonal components
since the order of gate applications is not successive. This can be
circumvented by resorting to schemes that produce sweeps with
ordered, successive gate applications \cite{ref:sornborger1999}.

A decomposition to order $p$ introduces an error of order
$\epsilon_{\delta t}=\llrr{\delta t}^{p+1}$. The error incurred in one
time step in general scales linearly with the system size~$N$.  This is
due to the nested commutators occurring in the error term of the
Suzuki-Trotter decomposition as can be seen when applying the
Baker-Campbell-Hausdorff formula. Since the number of time steps taken
is the total time $T$ divided by the number of time steps $T/\delta t$,
the total Trotter error $\epsilon_{trotter}$ is of order
$\mathcal{O}\llrr{\llrr{\delta t}^{p}NT}$ \cite{ref:gobert2005}.

The second considered error source is the truncation error. It stems
from the truncation of the Schmidt values during the application of a
two-site gate. Employing the Schmidt decomposition, a bipartite state
can be written as
\begin{align}
	\ket{\psi} &= \sum_{i=1}^{\chi '} \lambda_i
	\ket{\psi_i^{\text{left}}} \ket{\psi_i^{\text{right}}} +
	\sum_{i=\chi '+1}^{\chi} \lambda_i \ket{\psi_i^{\text{left}}} \ket{\psi_i^{\text{right}}} \nonumber  \\
	&= \ket{\psi_{\text{trunc}}} + \ket{\psi_{\bot}}
	\label{eq:bipartite}
\end{align}
where the left sum until $\chi '$ denotes the kept part
$\ket{\psi_{\text{trunc}}}$ and the right sum starting from $\chi '+1$
denotes the discarded part $\ket{\psi_{\bot}}$. Due to the fact
that the $\ket{\psi_i^{\text{left}}}$ are mutually orthogonal (as are
those of the right subsystem), the discarded part is orthogonal to the
retained. Given that the squared Schmidt values sum up to $1$, this
truncation leads to a deviation in the norm of the state
\begin{equation}
	\braket{\psi_{\text{trunc}}}{\psi_{\text{trunc}}} =
	1 - \sum_{i=\chi '+1}^{\chi} \lambda_i^2 = 1 - w
	\label{eq:defdiscweigt}
\end{equation}
where we defined the {\em discarded weight} $w=\sum_{i=\chi '+1}^{\chi}
\lambda_i^2$. Thus when renormalizing $\ket{{\psi_{\text{trunc}}}}$ we
pick up a factor of $1/\llrr{1-w}$. Thus upon $n$ truncations we are
off by a factor of about $\llrr{1-w}^{n_t}$ with $n_t$ being the number
of truncations performed. Truncating each bond in each time step
results in $n_t \propto \frac{NT}{\delta t}$ and thus the truncation
error is about
\begin{equation}
	\epsilon_{\text{trunc}} = \llrr{1-w}^{\frac{NT}{\delta t}} =
	\text{exp}\llrr{\frac{NT}{\delta t}\text{ln}\llrr{1-w}}
	\label{eq:truncerr}
\end{equation}

Thus we end up with a careful balancing of the two errors, depending on
the size of the time step $\delta t$. For smaller $\delta t$ we have a
smaller truncation error. Yet this requires more truncations due to the
larger number of time steps taken and thus in a larger truncation
error.

\section{An advanced application of the TEBD algorithm}
\label{sec:tedopa}

The TEBD algorithm is remarkably useful also in scenarios which at
first seem to be quite different from quantum many-body systems. One
such example is its usage in the time evolving density matrix using
orthogonal polynomials algorithm (TEDOPA) capable of treating open quantum
systems. We briefly present the TEDOPA scheme to show in which regimes
RRSVD proves to be most useful and how to speed-up previous simulations
and refer to \cite{ref:prior2010, ref:chin2010, ref:woods2014} for a
more detailed presentation of the algorithm.

TEDOPA is a certifiable and numerically exact method to treat open
quantum system dynamics \cite{ref:prior2010, ref:woods2014,
ref:woods2015}. It acts upon a spin-boson model description of an open
quantum system where a central spin interacts linearly with an
environment modelled by harmonic oscillators. In a two-stage process
TEDOPA then first employs a unitary transformation reshaping the
spin-boson model into a one-dimensional configuration. In a second step
this emerging configuration is treated by TEBD, exploiting its full
simulation power for one-dimensional systems.

The total Hamiltonian is split into system, environment and interaction
part
\begin{align}
	&H = H_{\text{sys}} + H_{\text{env}} + H_{\text{int}},
	\label{eq:H_sb1}\\
	&H_{\text{env}} = \int_0^{x_{\text{max}}}  \!\!\!\!dx~g\llrr{x}
			  a_x^\dagger a_x,
	\label{eq:H_sb2}\\
	& H_{\text{int}} = \int_0^{x_{\text{max}}}  \!\!\!\!dx~h\llrr{x}
			  \llrr{a_x^\dagger + a_x} A .
	\label{eq:H_sb3}
\end{align}

The bosonic creation and annihilation operators $a_x^\dagger$ and $a_x$
fulfill the usual bosonic commutation relations for the environmental
mode~$x$. The function $g\llrr{x}$ can be identified with the
environmental dispersion relation; the function $h\llrr{x}$ gives the
system-environment coupling strength for mode $x$ between its
displacement $\llrr{a_x^\dagger + a_x}$ and the operator $A$ acting on the
system.

Here the functions $g\llrr{x}$ and $h\llrr{x}$, together with the
temperature, uniquely characterize an environment and define the
spectral density $J\llrr{\omega}$ given by
\begin{equation}
	J\llrr{\omega} = \pi h^2\llrrq{g^{-1}\llrr{\omega}}
	\frac{dg^{-1}\llrr{\omega}}{d\omega}.
	\label{eq:def-sd}
\end{equation}

The interpretation of the quantity
$\llrr{dg^{-1}\llrr{\omega}/d\omega}\delta\omega$ is the number of
quantised modes with frequencies between $\omega$ and $\omega +
\delta\omega$ (for $\delta\omega\rightarrow0$). Further hold $g$ the
relations $g^{-1}\llrrq{g\llrr{x}} = g\llrrq{g^{-1}\llrr{x}} = x$.

Then new oscillators with creation and annihilation operators
$b_n^\dagger$ and $b_n$ can be obtained by defining the analytical
transformation~$U_n\llrr{x}$ as
\begin{align}
	&U_n\llrr{x} = h\llrr{x} p_n\llrr{x}, \\
	&b_n^\dagger =
	\int_0^{x_{\text{max}}} \!\!\!\!dx U_n\llrr{x} a_x^\dagger.
	\label{eq:chainmap}
\end{align}

It utilizes the orthogonal polynomials $p_n\llrr{x}$ defined with
respect to the measure $d\mu\llrr{x}=h^2\llrr{x}dx$. While in certain
cases it is possible to perform this transformation analytically, in
general a numerically stable procedure is used
\cite{ref:chin2010, ref:prior2010, ref:gautschi1994}.

This transformation yields a semi-infinite one-dimensional
nearest-neighbor Hamiltonian
\begin{align}
	H =&
	H_{\text{sys}} +
	t_0 A \llrr{b_0 + b_0^\dagger} +
	\sum_{n=1}^\infty \omega_n b_n^\dagger b_n  \nonumber \\
	& \hspace*{1.5cm} +\sum_{n=1}^\infty t_n
	\llrr{b_n^\dagger b_{n+1} + b_n b_{n+1}^\dagger}
	\label{eq:H_1D}
\end{align}

whose nearest-neighbor geometry (which is necessary for the
application of TEBD) as well as coefficients $\omega_n$ and
$t_n$ are directly related to the recurrence coefficients of the
three-term recurrence relation defined by the orthogonal polynomials
$p_n\llrr{x}$ with respect to the measure $d\mu\llrr{x}=h^2\llrr{x}dx$
\cite{ref:chin2010}.

\begin{figure}[hbt]
	\begin{center}
	\includegraphics[width=0.9\columnwidth]{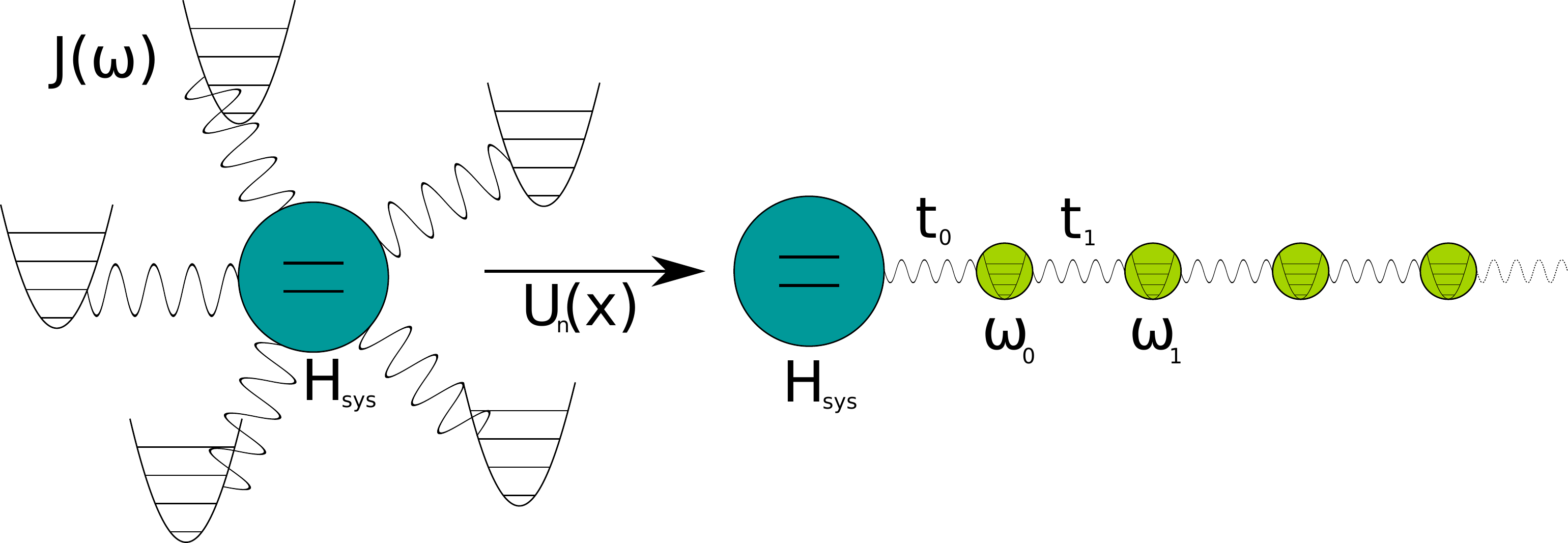}
	\caption{Illustration of the spin-bonson model's transformation
	into a one-dimensional configuration where the system is only
	coupled to the environment's first site.}
	\label{fig:tedopa}
	\end{center}
\end{figure}

This transformation of the configuration is depicted in Fig.~\ref{fig:tedopa},
from the spin-boson model on the left to a one-dimensional geometry on the right.
In a last step it is necessary to adjust this configuration further to suit
numerical needs. The number of levels for the environment's oscillator on
site~$k$ can be restricted to $d_{k\text{,max}}$ to reduce required computational
resources. A suitable value for $d_{k\text{,max}}$ is related to this
site's average occupation, depending on the environment's structure and
temperature. The number of sites that is required to faithfully represent
the environment has to be sufficiently large to completely give the appearance of
a ``large'' reservoir - one that avoids unphysical back-action on the system
due to finite-size effects that lead to reflections at the system boundaries
(see \cite{Rosenbach15} for extensions of TEDOPA that can alleviate this
problem considerably). It should be noted that these truncations, while feasible
numerically and easily justifiable in a hand-waving kind of argument, in the end
can also be {\em rigorously certified} by analytical bounds \cite{ref:woods2015}.
These adjustments yield an emerging system-environment configuration which is
now conveniently accessible by TEBD.

%
% Further explain g, h, and A

%This model uniquely defines the system and the environment. The latter
%is given by $g$ and $h$.

%The gist of TEDOPA is: take the defined model
%Eqs.\eqref{eq:H_sb1}--\eqref{eq:H_sb3}, transform it into a
%one-dimensional configuration and treat this configuration by t-DMRG.

%
\section{Reduced-rank Randomized SVD} \label{sec:RRSVD}
The Singular Value Decomposition (SVD) is at the heart of  the MPS representation and MPS-based algorithms,
such as TEBD. The efficiency of TEBD comes from the possibility of approximating states living in an
exponentially large Hilbert space with states defined by a number of parameters that grows only polynomially
with the system size. In order to understand why the SVD plays such a crucial role, we introduce the following
problem: given a complex $m \times n$ matrix $A$, provide the best rank-$k$ ($k\leq n$) approximation of $A$.
Without loss of generality we suppose  $m \geq n$ and  $rank(A)=n$. The solution to this problem is well known
\cite{golub96}: first compute the Singular Value Decomposition of $A = U\Sigma V^\dagger$, where
$U=\llrr{U^{(1)},U^{(2)},\ldots,U^{(n)}}$, $V=\llrr{V^{(1)},V^{(2)},\ldots,V^{(n)}}$ are the left and right
singular vectors of $A$ respectively and $\Sigma = diag\llrr{\sigma_1,\sigma_2,\ldots,\sigma_n}$ with
$\sigma_1 \geq \sigma_2 \geq  \ldots \geq \sigma_n$. Then retain the first $k$ largest singular values
$\llrr{\sigma_1,\sigma_2,\ldots,\sigma_k}$ and  build the matrices $U_k = \llrrq{U^{(1)},U^{(2)},\ldots,U^{(k)}}$
and $V_k = \llrrq{V^{(1)},V^{(2)},\ldots,V^{(k)}}$. The matrix $\widetilde{A}_k = U_k \Sigma_k V_k^\dagger$ satisfies
\begin{align}
    ||A-\widetilde{A}_k||_F = \sqrt{\sum_{i=k+1}^{n}  \sigma_i^2} = \min_{rank(A') = k}||A-A'||_F  ,
\end{align}
where $|| \cdot ||_F$ indicates the Frobenius norm
\begin{align}\label{eq:frobNorm}
    ||A||_F = \sqrt{\sum_{i=1}^{n}  \sigma_i^2}.
\end{align}
In other words, $\widetilde{A}_k$ provides the best rank-$k$ approximation of $A$. This result justifies the
application of SVD for the TEBD decimation step.

The computational complexity of the SVD of $A$ is $\mathcal{O}(m \cdot n^2)$.
For large matrices, the SVD can therefore require a significant amount of time. This is a crucial point
since every single TEBD simulation step of an $N$-sites spin chain requires $\mathcal{O}\llrr{N}$ SVDs
which usually consumes about 90\% of the total simulation time.

As discussed in \sref{sec:tebd}, the bond-dimension $\chi$ requires discarding of $n-\chi$ singular values
(and corresponding left-/right- singular vectors). In the TEBD two-site update step, for example, we keep
only $\chi$ singular values out of  $n=d \cdot \chi$ (pure states) or $n=d^2 \cdot \chi$ (mixed states).
Most of the singular values and left-/right-singular vectors are therefore discarded. It means that we are
investing time and computational resources to compute information that is then wasted.

It is possible to avoid the \emph{full} SVD of $A$ and compute only the first $k$ singular values
and corresponding singular vectors by using \emph{Truncated SVD} methods; such methods are standard tools
in data-classification algorithms \cite{hastie09}, signal-processing \cite{candes09} and other research fields.
The  Implicitly Restarted Arnoldi Method \cite{sorensen98,sorensen02} and the Lanczos-Iteration \cite{larsen98}
algorithms, both belonging to the Krylov-subspace iterative methods \cite{saas92,hochbruck97}, are two examples.

The Reduced-Rank Singular Value Decomposition (RRSVD), originally presented in by N. Halko \emph{et al.}
\cite{halko11}, is a \emph{randomized} truncated SVD. It is particularly well suited to decompose structured
matrices, such as those appearing in the TEBD simulation of non-critical quantum systems. Most interestingly,
the algorithm is insensitive to the quality of the random number generator used, delivers highly accurate results
and is, despite its random nature, very stable: the probability of failure can be made arbitrarily small with
only minimal impact on the necessary computational resources.\\ In what follows we will describe the RRSVD
algorithm and report the main results on stability and accuracy. For a full account on RRSVD we refer the reader
to \cite{halko11}.

The RRSVD algorithm is a two-step procedure. The first step constructs an orthogonal matrix $Q$ whose columns 
constitute a basis for  the approximated range of the input matrix $A$. In the second step, the approximated 
SVD of $A$ is computed by performing a singular value decomposition of $Q^\dagger A$.

The approximation of the range can be done either for  fixed  \emph{error tolerance} $\epsilon$, namely by finding a  $Q_\epsilon$ such that
\begin{align} \label{eq:fixerrprob}
    ||A-Q_\epsilon Q_\epsilon^\dagger A|| \leq \epsilon,
\end{align}
or for a fixed rank $k$ of the approximation, that is by finding $Q_k$ such that
\begin{align}
    \min_{rank(X)\leq k} || A- X|| \approx ||A-Q_kQ_k^\dagger A||.
\end{align}
The first one is known as the \emph{fixed-precision approximation problem}, whereas the second is known as the \emph{fixed-rank approximation problem}.
Here and in what follows, we indicate by $||A||$  the operator norm, corresponding to the largest singular value of $A$.

If $\sigma_j$ is the $j$-th largest singular value of $A$ then
    \begin{align}
    & \min_{rank(X)\leq k(\epsilon)} ||A-X||_F  = \nonumber \\ & = ||A - Q_{k(\epsilon)}Q_{k(\epsilon)}^\dagger A||_F =
    \sqrt{\sum_{j=k(\epsilon)+1}^n \sigma_j^2}.
\end{align}
where $\sigma_{k(\epsilon)}$ is the first singular value $\geq \epsilon$ and the columns of the rank-$k(\epsilon)$ matrix $Q$ are the first $k(\epsilon)$ left singular vectors of $A$.
However,  this would require the knowledge of the first $k(\epsilon)$ singular values and vectors of $A$.

For the sake of simplicity, let us focus initially on the fixed-rank approximation problem.
In order to determine $Q_k$, we resort to randomness.
More precisely, we use a sample of $k+p$ random vectors $\omega^{(i)}$, whose  components are independently drawn form a standard Gaussian distribution $\mathcal{N}_{(\mu=0,\sigma=1)}$.
The set  $\left \{ \omega^{(i)} \right \}_{i=1}^{k+p}$ will be with \emph{very high probability} a set of independent vectors.
The parameter $p$ determines the amount of \emph{oversampling} needed to make the rank-$k$ approximation of $A$ more precise.
The $m \times (k+p)$ matrix
\[
    Y=A\Omega,
\]
where $\Omega = (\omega^{(1)},\omega^{(2)},\ldots,\omega^{(k+p)})$, will therefore have full rank $(k+p)$.
By re-orthogonalizing the columns of the matrix $Y$ we  obtain a basis for the rank-$(k+p)$ approximation of the range of $A$.
The re-orthogonalization can be  done by using the $QR$-decomposition $Y=Q R$.
If $(k+p) < n$, the computational cost of the first step of this algorithm is dominated by the matrix multiplication $A\Omega$: this operation has an asymptotic complexity  $\mathcal{O}(mn(k+p))$, whereas the QR decomposition of the $m \times (k+p)$ matrix $Y$ has asymptotic complexity $\mathcal{O}(m(k+p)^2)$ .

When the input matrix $A$ is very large, the singular vectors associated with small singular values may interfere with the calculation.
In order to reduce their weight relative to the dominant singular values it is expedient to take powers of the original matrix $A$.
So, instead of computing $Y=A\Omega$ we compute
\begin{align} \label{eq:Z}
    Z =B\Omega =  (AA^\dagger)^qA\Omega.
\end{align}
The singular vectors of $B =(AA^\dagger)^q A  $ are the same as the singular vectors of $A$; for the singular values of $B$, on the other hand, it holds:
\begin{align} \label{eq:powSV}
    \sigma_j(B) = \llrr{\sigma_j(A)}^{2q+1},
\end{align}
which leads to the desired reduction of the influence on the computation of the singular vectors associated to small singular values.
This ``stabilizing'' step, also referred to as \emph{Power Iteration} step (PI), increases the computational cost of the first step of the RRSVD algorithm by a constant factor $(2q+1)$.

A side effect of the PI is the extinction of all information pertaining to singular vectors associated to small singular values due to the finite size of floating point number representation.
In order to avoid such losses, we use intermediate re-orthogonalizations (see Algorithm \ref{alg:RSI}).
We point out that in the typical context where TEBD is successfully applied, the singular values of the according matrices decay very fast, so the PI scheme must be applied.

Since the RRSVD is a randomized algorithm, the quality of the approximation comes in the form of expectation values, standard deviations and failure probabilities.
We report pertinent results that can be found (including proofs) in \cite{halko11}, as well as some particular results tuned to cases of specific relevance to applications for TEBD.
\begin{theorem} \label{th:errexpiter}
(Corollary 10.10 in \cite{halko11})
Given a target rank $k \geq 2$ and an oversampling parameter $p\geq2 $, let $Q_Z$ be the orthogonal $m \times (k+p)$ matrix consisting of the fist $k+p$ left singular vectors of the matrix $Z$ defined in \eref{eq:Z}.
We define $P_Z = Q_Z^ \dagger Q_Z$. Then
\begin{align}\label{eq:errexpiter}
    &\mathbb{E}\llrr{||\llrr{\mathbb{I} - P_Z }A}||  \leq   \\
    & \leq \llrrq{\llrr{1+\sqrt{\frac{k}{p-1}}} \sigma_{k+1}^{2q+1} +\frac{e \sqrt{k+p}}{p} \llrr{\sum_{j>k} \sigma_j^{2\llrr{2q+1}}}^{\!\!\frac{1}{2}}}^{\frac{1}{2q+1}} \nonumber \\
    & \leq \llrrq{ 1+\sqrt{\frac{k}{p-1}}+\frac{e \sqrt{k+p}}{p} \sqrt{n-k}  }^{1/\llrr{2q+1}} \sigma_{k+1}. \nonumber
\end{align}
\end{theorem}
The application of the PI scheme reduces the average error exponentially in $q$.
In order to quantify the deviations from the expected estimation error we use the following facts:
\begin{align} \label{eq:relAB}
    || \llrr{\mathcal{I} - P_Z}A||^{2q+1} \leq || \llrr{\mathcal{I} - P_Z}B||
\end{align}
(Theorem 9.2 \cite{halko11})  and
\begin{align} \label{eq:perr}
    &|| \llrr{\mathcal{I} - P_Y}A|| \leq \\
    &1+6 \sqrt{(k+p) \cdot p \log(p)}  \sigma_{k+1} +3 \sqrt{k+p} \llrr{\sum_{j>k} \sigma_j^2}^\frac{1}{2}  , \nonumber
\end{align}
with probability greater or equal to $1- \frac{3}{p^p}$ (Corollary 10.9 \cite{halko11}).
We can now state the following
\begin{corollary} \label{th:tamaErr}
Under the same hypotheses of Theorem \ref{th:errexpiter} it is
\begin{align}
P \llrr{|| \llrr {\mathcal{I} - P_Z}A||  \leq \alpha^\frac{1}{2q+1} \sigma_{k+1}  + \beta^\frac{1}{2q+1}  \sum_{j>k}\sigma_j} \geq 1- \frac{3}{p^p},
\end{align}
with $\alpha = (1+6 \sqrt{(k+p) \cdot p \log(p)})$ and $\beta = 3 \sqrt{k+p}$.
\end{corollary}
\begin{proof}
By applying  \eref{eq:perr} to $B =(AA^\dagger)^q A $, and using \eref{eq:powSV} we have
\[
P \llrr{|| \llrr{\mathcal{I} - P_Z}B|| \leq\alpha \sigma_{k+1}^{2q+1} +\beta \llrr{\sum_{j>k} (\sigma_j^{2q+1})^2}^\frac{1}{2}       } \geq 1- \frac{6}{p^p}.
\]
Using the relation \eref{eq:relAB} we have that
\[
    || \llrr{\mathcal{I} - P_Z}A|| \leq \llrr{\alpha \sigma_{k+1}^{2q+1} +\beta \llrr{\sum_{j>k} (\sigma_j^{2q+1})^2}^\frac{1}{2}    }^{1/(2q+1)}
\]
since the function $f(q) = a^{1/x}, a>0, x>0$ is convex, it holds
\begin{align}
    &\llrr{\alpha \ \sigma_{k+1}^{2q+1} +\beta \llrr{\sum_{j>k} (\sigma_j^{2q+1})^2}^\frac{1}{2}    }^{1/(2q+1)}  \nonumber \\
    & \leq  \llrr{\alpha \  \sigma_{k+1}^{2q+1}}^\frac{1}{(2q+1)} + \llrr{\beta \sqrt{\sum_{j>k} (\sigma_j^{2q+1})^2}}^\frac{1}{2q+1} \nonumber \\
    &\leq  \alpha^\frac{1}{2q+1} \sigma_{k+1} +\beta^\frac{1}{2q+1}\sqrt{ \sum_{j>k} \sigma_j^2}  \nonumber \\
    &= \alpha^\frac{1}{2q+1} ||A - \widetilde{A}_k|| +\beta^\frac{1}{2q+1}||A - \widetilde{A}_k||_F.
\end{align}
\end{proof}
The results provided by the algorithm are usually closer to the average value than those estimated
by the bound \eref{eq:errexpiter}, which therefore seems to be not tight. However, the important
message of the preceding results is that by applying PI we obtain much better approximations of $A$
than those provided by the original scheme while error expectations and deviations are under full control.
%%%%%%%%%%%%%%%%%%%%%%%%%%%%%%%%%%%%%%%%%%%%
%%%%%%%%%%%%%%%%%%%%%%%%%%%%%%%%%%%%%%%%%%%%
%%%%%%%%%%%%%%%%%%%%%%%%%%%%%%%%%%%%%%%%%%%%
\begin{algorithm}
\caption{\bf{Randomized SVD with Power Iterations}}
\label{alg:RSI}
\begin{algorithmic}[1]
%\Procedure{}{}
\Require $m \times n$ matrix $A$; integers $l = k+p$ (rank of the approximation) and $q$ (number of iterations).

\State Draw an $ n \times l$ Gaussian matrix $\Omega$.
\State Form $Y_0 = A \Omega$.
\State Compute the QR factorization $Y_0 = Q_0 R_0$.
\For {$j=1, 2,\dots, q$}
	\State Form $\widetilde{Y}_j = A^\dagger Q_{j-1} $.
	\State Compute the QR factorization $\widetilde{Y}_j = \widetilde{Q}_j \widetilde{R}_j$.
	\State Form $Y_j = A \widetilde{Q}_{j} $.
	\State Compute the QR factorization $Y_j = Q_j R_j$.
\EndFor

\State \Return $Q = Q_q$.
%\EndProcedure
\end{algorithmic}
\end{algorithm}
%%%%%%%%%%%%%%%%%%%%%%%%%%%%%%%%%%%%%%%%%%%%
%%%%%%%%%%%%%%%%%%%%%%%%%%%%%%%%%%%%%%%%%%%%
%%%%%%%%%%%%%%%%%%%%%%%%%%%%%%%%%%%%%%%%%%%%
It would be most convenient to have some means to check how close $QQ^\dagger A$ is to the original input matrix $A$.
With such a tool, we could check the quality of the approximation; moreover, we would be able to solve the fixed-error approximation problem \eref{eq:fixerrprob}.
In the TEBD setting, this would allow us to determine the bond dimension $\chi$ for an assigned value $\epsilon$ of the truncation error.
The solution to this problem comes from this result:
\begin{theorem} \label{th:accCheck} (equation 4.3 supported by Lemma 4.1 \cite{halko11}) With $M=(I-QQ^\dagger)A$
\begin{align}
    &P\llrr{||M|| \leq 10 \sqrt{\frac{2}{\pi}} \max_{i=1,2,\ldots,r} ||M \omega^{(i)}||} \geq 1-10^{-r},
\end{align}
where $\omega^{{i}},i=1,\ldots,r$ are standard normal random vectors.
\end{theorem}
Suppose that we have completed the first three steps of Algorithm \ref{alg:RSI}.
Set $Q = Q_0$, with $rank(Q) = l=k+p$ and choose the size $r$ of the sample.
The \emph{Accuracy Check}  algorithm (Algorithm \ref{alg:RAC}) takes in input $A,Q,r$ and $\epsilon$ and returns a new matrix $Q'$ that satisfies the accuracy bound with probability $1-10^{-r}$.
%%%%%%%%%%%%%%%%%%%%%%%%%%%%%%%%%%%%%%%%%%%%
%%%%%%%%%%%%%%%%%%%%%%%%%%%%%%%%%%%%%%%%%%%%
%%%%%%%%%%%%%%%%%%%%%%%%%%%%%%%%%%%%%%%%%%%%
\begin{algorithm}[h]
\caption{\bf{Accuracy check}}
\label{alg:RAC}
\begin{algorithmic}[1]
%\Procedure{}{}
\Require $m \times n$ matrix $A$; rank-$k+p$ projector $P_Q=QQ^\dagger$; integer $r$;  tolerance $\epsilon$.
\Do
\State Set $l=k+p$.
\State Draw an $ n \times r$ Gaussian matrix $\Omega_r$.
\State Compute $B=A\Omega_r$.
\State Compute $D =(I-P_Q) A \Omega_r = \llrr{d^{(1)}, d^{(2)},\ldots,d^{(r)}}$.
\State Set MAX= $\max\left \{ d^{(i)} \right \}_{i=1}^r$.
\If{$( \text{MAX} > \epsilon) $}
\State Build $\widetilde{Q} = Q|B$.
\State Set $l = l+r$.
\State Compute the QR decomposition $\widetilde{Q} = \widetilde{Q}'\widetilde{R}'$.
\State Set $Q = \widetilde{Q}'$.
\EndIf
\doWhile($\text{MAX} >\epsilon$ and $l \leq n-r$ )
\State \Return $Q$.
%\EndProcedure
\end{algorithmic}
\end{algorithm}
%%%%%%%%%%%%%%%%%%%%%%%%%%%%%%%%%%%%%%%%%%%%
%%%%%%%%%%%%%%%%%%%%%%%%%%%%%%%%%%%%%%%%%%%%
%%%%%%%%%%%%%%%%%%%%%%%%%%%%%%%%%%%%%%%%%%%%

The computational cost of the Accuracy Check depends on different parameters.
The cost of each iteration step is $\mathcal{O}(m \cdot l^2)$.
Then we have to consider the iterations.
If we take $r$ too small (e.g. $r=1$), then if the starting rank-$l$ approximation is not good enough,
we might need many iteration to converge to the desired accuracy. As a rule of thumb, we suggest to
double the rank of the approximation at each time. This will likely lead to oversampling, but still
delivers a good performance balance.

Since the reference metric in TEBD is the Frobenius norm  eq. \eref{eq:frobNorm}, some estimate of
the Frobenius norm via the operator norm is required. To this end we observe that the TEBD is successful
when the correlations in the simulated system are sufficiently short-ranged, i.e. the ``singular values decay fast enough''.
If the entanglement between distant parts of the system is non-vanishing, the decimation step will lead
to an important loss of relevant information. As an example let us consider a spin-chain of
size $n$ and a bipartition $A=\{1,2, \ldots l\}, \ B = \{l+1,\ldots, n\} $ of the chain. Let $  \ket{k}_A$
and $\ket{j}_B$ be orthonormal bases for subsystems $A$ and $B$ respectively. Then the spin-chain state
\ket{\psi} can be Schmidt-decomposed as
\[
    \ket{\psi} = \sum_{k,j} c_{k,j} \ket{k}_A \ket{j}_B = \sum_i \sigma_i \ket{i}_A \ket{i}_B,
\]
where in the last equality we used the SVD decomposition of the matrix $ C= \left ( c_{j,k}\right ) = U \Sigma V^\dagger$
to perform the Schmidt decomposition of \ket{\psi}. The Schmidt coefficients $\sigma_i$ are the singular values of
$C$, i.e. the diagonal elements of $\Sigma$ \cite{nielsen11}. The number of non-zero singular values is the Schmidt number.
The amount of entanglement between the subsystems $A$ and $B$ can be quantified by the von Neumann, or entanglement, entropy
\cite{PlenioVirmani2007} of the reduced density matrices $\rho_A$ and $\rho_B$
\begin{align}
    S(\rho_A) = -\sum_i \sigma_i^2 \log(\sigma_i^2) = S(\rho_B).
\end{align}
The decay rate of the singular value is therefore directly related to the amount of entanglement
shared between two parts of a system. If the system is highly entangled, the singular values will
decay ``slowly''; in the limiting case where the system is maximally entangled, the reduced density
matrices will describe completely mixed states and the singular values will be all equal to each other.
When the system is only slightly entangled, the singular values will decay very fast; in particular,
if the state \ket{\psi} is separable, the Schmidt number is equal to 1 (and $\sigma_1 = 1$).
The behavior of the entanglement entropy in systems at, or close to, the critical regime has been
thoroughly studied (see \cite{eisert08} and references therein) together with its dependence on the
decay rate of the eigenvalues of the reduced density matrices $\rho_{A,B}$ \cite{calabrese08, schuch08, horodecki13, horodecki15}.
By oversimplifying the problem, we observe that if the singular values decay as $\sigma_j = 1/\sqrt{j}$
the entanglement entropy shows, when the system size  $n \to \infty$, a divergence $\log^2(n)$,
whereas if they decay as $\sigma_j = 1/j$ the entanglement entropy does converge, since an area law
holds \cite{eisert08}.
When the singular values decay as $1/j$, however,  the truncation error $||A - \widetilde{A_k}||_F$ decreases slowly in $k$ and  the TEBD scheme becomes inefficient since any decimation will lead to considerable approximation errors (see Eq.~\eref{eq:truncerr}).
For these reasons we consider the case of linearly decreasing singular values as an extremal case for the range of applicability of the TEBD scheme.

This observation provides a useful tool to estimate the Frobenius norm, which plays a central role in TEBD,  through the operator norm.
For any matrix $A$ it is:
\begin{align}
||A|| \leq ||A||_F \leq \sqrt{rank(A)} ||A||. \nonumber
\end{align}
This result holds in general.
The inequalities are saturated, in particular,  when $rank(A)=1$.
The upper bound on the Frobenius norm is the tighter the closer the singular values of $A$ are  to each other.
If the singular values decay at least as $1/j$, on the other hand, we have the following result.
\begin{fact}
Given a rank-$n$ matrix $A$ with singular values $\sigma_1 \geq \sigma_2 \geq \ldots \geq \sigma_n$ with  $\sigma_j \leq  \sigma_1/j, j=1,\ldots,n$, it holds:
\begin{align}
    ||A||_F \leq \frac{\pi}{\sqrt{6}}  ||A|| = \sigma_1 \frac{\pi}{\sqrt{6}}
\end{align}
where $||A||$ indicates the operator norm.
\end{fact}
\begin{proof}
\begin{align}
    ||A||_F = \sqrt{Tr\llrr{A^\dagger A}} = \sqrt{\sum_{i=1}^n \sigma_i^2} \le \sigma_1 \sqrt{\sum_{i=1}^n \frac{1}{i^2}}.
\end{align}
The last term is upper bounded by
\[
\sigma_1 \lim_{n \to \infty} \sqrt{\sum_{i=1}^n \frac{1}{i^2}} = \sigma_1 \frac{\pi}{\sqrt{6}}.
\]
\end{proof}
This result finds application in the Accuracy Check routine. If we set $P_Q= QQ^\dagger$ we find
\begin{align}
    ||(I -P_Q)A||_F  &= \sqrt{Tr\llrr{\llrr{A-P_Q A}^\dagger \llrr{A-P_Q A}}} \nonumber \\
    & =  \sqrt{Tr \llrr{A^\dagger A} - \llrr{A^\dagger P_Q A}} \nonumber \\
    &= \sqrt{||A||_F^2 - ||A^\dagger Q||_F^2}\nonumber \\ 
    & \leq ||A||_F.
\end{align}
Therefore
\begin{align}
    ||(I-P_Q)A||_F &\leq \frac{\pi}{\sqrt{6}} ||A||_F  \\
    & \leq \frac{10}{3} \pi   \max_{i=1,2,\ldots,r}||(I-QQ^\dagger)A \omega^{(i)}|| \nonumber
\end{align}
with probability $1-10^{-r}$.

We point out that it is  not really necessary to \emph{estimate} the Frobenius norm of the error: given $Q$ we can compute $ ||(I-QQ^\dagger)A||_F$ directly.
However one should note that this computation of $QQ^\dagger A$ requires $2 (m \cdot n \cdot k)$ float operations instead of the $ 2 (m \cdot k \cdot r)+ m \cdot n \cdot r$ operations required to get the error estimate.% We observe that the error probability $1-10 ^{-r}$, becomes negligible ($10^{-50}$ or much much less) for every real-life situation (remind that this algorithm is designed to work on large ($>500\times500$) matrices).

Now that we have the orthogonal matrix $Q$ whose columns constitute a basis for the approximated range of the input matrix $A$, we can directly compute the approximate SVD of $A$ by:
\begin{enumerate}[i)]
\item Form the matrix $B = Q^\dagger A$.
\item Compute the SVD of $B$: $B=\widetilde{U}\Sigma V^\dagger$.
\item Form the orthonormal matrix $U = Q\widetilde{U} $.
\end{enumerate}
The product $Q^\dagger A$ requires $\mathcal{O}((k+p) \cdot n \cdot m)$ floating point operations.
The SVD of $B$, under the reasonable assumption that $k+p\leq n$,  requires $\mathcal{O}(n \cdot (k+p)^2)$ operations.
The product $Q\widetilde{U}$ requires $\mathcal{O}(m \cdot (k+p)^2)$ operations.

We conclude this section by presenting some results on the computational cost of the whole 
procedure and on some more technical aspects related to the implementation of the algorithm.
The asymptotic computational cost of partial steps has been given at various places in this section.
In summary, the real bottleneck of the complete RRSVD algorithm is the first matrix multiplication $Y=A \Omega$ which has complexity $\mathcal{O}(m \cdot n \cdot(k+p))$.
The value of $p$ can be set in advance or determined by the Accuracy Check method described above.
All remaining operations, such as QR decompositions and error estimation, have smaller computational complexity.
If the TEBD scheme is applied in a non-adaptive way, i.e. the bond dimension is kept fixed at a given value $\chi$,  we use RRSVD to solve the fixed-rank problem.
In this case RRSVD has complexity $\mathcal{O}(m\cdot n  \cdot \chi)$.
If  the bond dimension $\chi$ is set independently of the input matrix size, the replacement of the standard SVD by RRSVD will  therefore result in a speed-up linear in $n$.
On the other hand, if we use an adaptive TEBD simulation where the bond dimension is set such that some bound on the approximation error is satisfied, the cost of RRSVD will (strongly) depend on the structural properties of the input matrices.
If the singular values decay exponentially (short-range correlations), then the expected speed-up is roughly the same as for the non-adaptive scheme.
If the simulated system presents long-range correlations, then the speed-up provided RRSVD will be less then linear in $n$, possibly even vanishing.
However, TEBD itself is not the ideal tool to deal with systems exhibiting long-range correlations, so this is only a minor limitation.

Another crucial observation, related to the implementation, is that all the operations required by RRSVD are standard functions of either the Basic Linear Algebra Subprograms (BLAS) \cite{blas} or the Linear Algebra PACKage (LAPack  \cite{lapack}) libraries.
Both libraries are heavily optimized and are available for single-core, multi-core (e.g. Intel Math Kernel Library (MKL) \cite{mkl}) and Kilo-processor architectures (e.g. CuBLAS \cite{cublas} and CULA \cite{cula}).
Since TEBD simulations are usually run on many cores (on current cluster architectures often 8 or 16), RRSVD can take full advantage of the optimized libraries.
\section{RRSVD case-study and profiling of real TEBD simulations}
\label{sec:rrsvdtedopa}
We start by showing that, despite its random nature, RRSVD produces very accurate  results with surprisingly small fluctuations.
To this end we test the described algorithm on a  sample of relatively small structured matrices extracted from pure-state TEBD simulations of our standard system from section~\ref{sec:tedopa}, subsequently continuing to larger matrices from mixed-state TEBD simulations of the same system.
We analyze the accuracy of RRSVD and its time performances, concluding the section by  presenting  how RRSVD impacts full TEBD simulations. %We include the Accuracy Check: we ask RRSVD to increase the rank of the approximation as soon as the estimation of the distance between $A$ and its rank-$k$ approximation is larger than the threshold value $\epsilon$.

\subsection{Stability analysis} \label{sec:stability}
In order to benchmark the stability of the RRSVD algorithm we consider a set of $7$ diagonal $750 \times 750$ matrices $\Sigma_m,\ m=1,2,\ldots,7$.
The matrices $\Sigma_m$ are extracted from an actual pure-state TEBD simulation of our benchmark system described in section \ref{sec:tedopa}.
For every $\Sigma_m$ we generate a set of $n_A=20$ random matrices $\left \{ A_{m,i}\right \}_{i=1}^{n_A}$ by randomly generating $1500 \times 750$ random orthonormal matrices $U_{m,i}$ and $750 \times 750$ random orthonormal matrices $V_{m,k}$; each $A_{m,i}$ has dimensions $1500 \times 750$.
In this way we check for the dependence of both LaPack SVD and RRSVD on different matrices with the same structural properties.
In order to take the random nature of RRSVD into account, for \emph{each}  $A_{m,i}$ we perform $n_R=20$ executions of RRSVD.
In this first benchmark we provide a rank-$50$ approximation of the rank-$750$ original matrices $A_{m,i}$ and show how the 
accuracy is related to the number of subspace iterations $q$. Motivated by the theorems reported in the previous section, we 
set $p=50$. We compare the accuracy and timing results for different values of the iteration parameter $q=2,4,6$. The \emph{accuracy-check} 
part of the algorithm is not included here: we do not estimate the difference between the original matrix $A$ and its projection on the reduced space.

By running the RRSVD on a set  $\left \{A_{k,i} \right \}_{i=1}^{n_A}$ random realizations of matrices exhibiting the same structure, i.e. same singular values $\Sigma_k$, we check that the accuracy of RRSVD depends only, for fixed numbers of iterations $q$, approximation-rank $k$ and oversampling parameter $p$,  on the structural properties of the matrices.
Therefore, in what follows we present an analysis referring to the instance corresponding to the random realization $A_{2,1}$, which is, in every respect, a good representative of the kind of matrices we deal with when performing a TEBD simulations of a pure quantum system far from the critical regime.
In  \fref{fig:figure1a} we plot the singular values of $A_{2,1}$.
It is important to notice that some of the largest singular values are very close to each other.
This is a typical situation in which truncated-SVD methods belonging to the family of Krylov-subspace iterative methods are likely to require more iterations in order to accurately resolve the singular values.
RRSVD, on the other hand, is completely insensitive to this peculiarity of the spectrum and provides very good results for the whole range of retained singular values ($k=50$) starting from $q=4$: the approximation error is comparable to that of the state-of-the-art LAPack SVD routine.
Most noticeably, none of the $n_R$ executions of RRSVD on the instance $A_{2,1}$ presents \emph{outliers}, that is to say singular values that are computed with anomalous inaccuracy (\fref{fig:figure1b}).
The behavior of the approximation error as a function of $q$ is compatible with the theoretical results stated in the previous section.
In Table \ref{tab:tabSpeedUp1} we show the speed-up $t_{SVD}/t_{RRSVD}$ when both, MKL-SVD and MKL-based RRSVD (see section \ref{sec:profiling} for more information about the implementation), are executed on an Intel Xeon X5570@2.93GHz by 8 concurrent threads.
The speed-up over the MKL-SVD is obviously decreasing as $q$  and $k$ increase: for $k=p=100$ and $q=6$ almost no advantage remains in applying RRSVD instead of the standard SVD.
\begin{table}[t]
\begin{center}
\begin{tabular}{||c|| | c | c | c | c | c | c ||}
\hline \hline
 k/q & 0 & 2 & 4 & 6  & 8 & 10\\ \hline \hline
 50 &  11.6 & 5.4 & 3.5 & 2.6 & 2.04 & 1.7 \\ \hline
 100 &  4.7 & 2.3 & 1.5 &1.1 & 0.89 & 0.69 \\ \hline \hline
\end{tabular}
\end{center}
\caption{$A_{2,1}$: RRSVD Speed-up $t_{SVD}/t_{RRSVD}$.  LAPack SVD time: 3.84 s}
\label{tab:tabSpeedUp1}
\end{table}%
It is worth stressing here that RRSVD is meant to deliver a substantial advantage only for very large matrices and a comparatively small number of retained dimensions, as we will show later.
\begin{figure}[h]
\subfigure[]{\label{fig:figure1a} \includegraphics[width=\columnwidth]{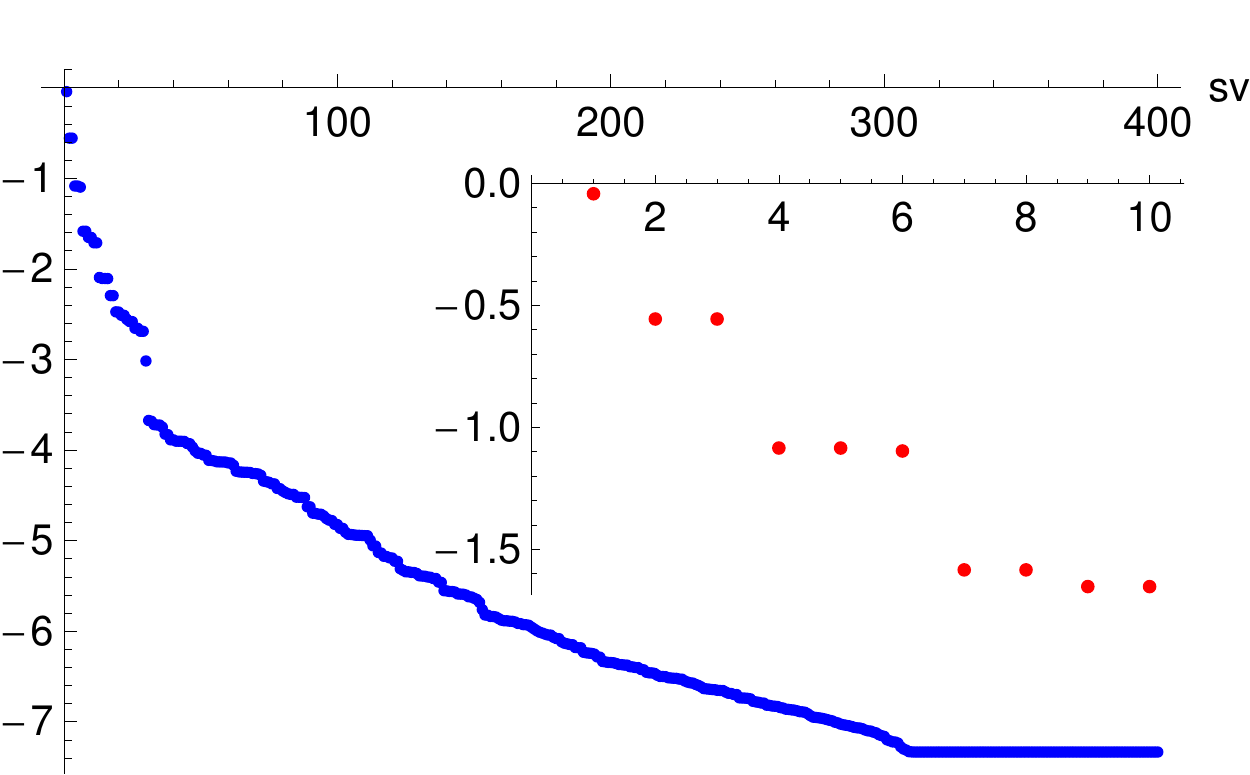}}
\subfigure[]{\label{fig:figure1b} \includegraphics[width=\columnwidth]{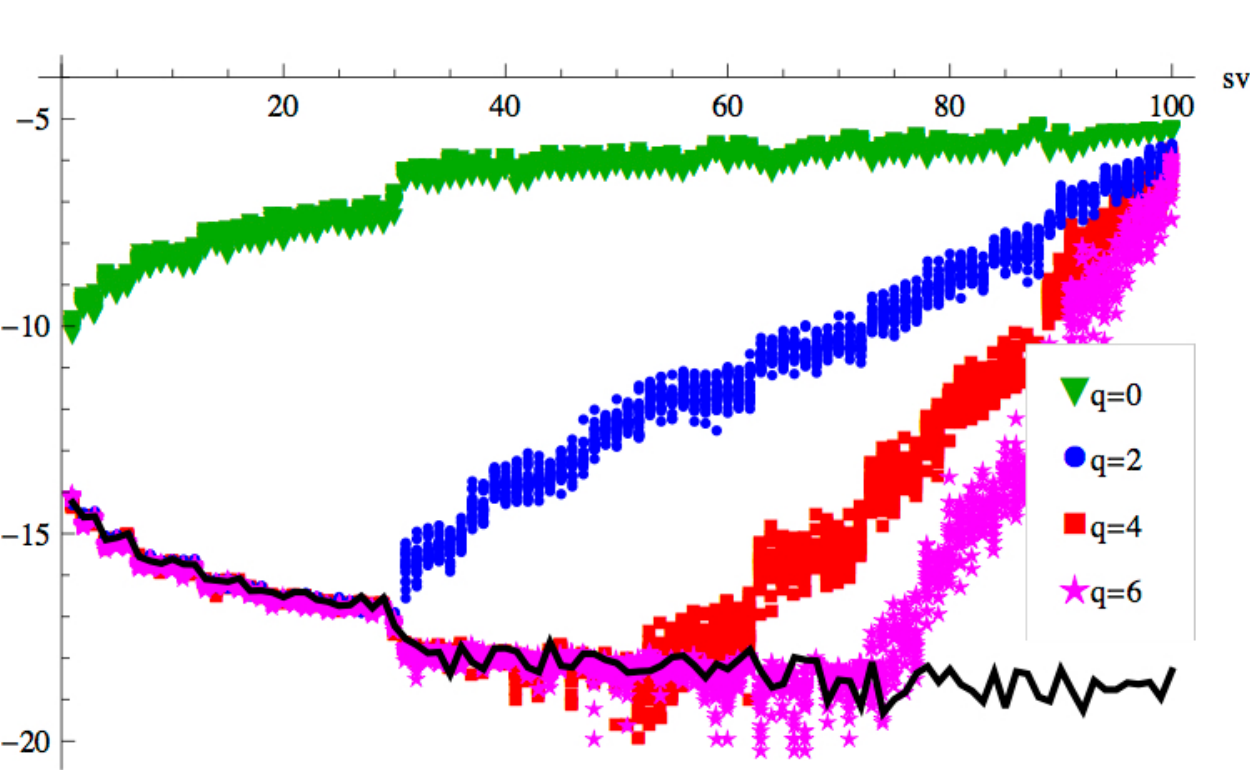}}
\caption{Instance $A_{2,1}$: $m=1500, n=750$; $k=p=50$.
Discarded weight: $w=4 \cdot 10^{-4}$ (a)Base-10 logarithmic plot of  the singular values $\Sigma_2$: the decay appears roughly exponential.
The inset shows the first $10$ singular values: some of these SVs are very close to each other (about $10^{-5}$) (b)  The errors $\log_{10}(|\sigma_i - \sigma_i^{RRSVD}|)$ of the RRSVD  for each singular value and for all  $n_R=20$ executions of RRSVD on the same instance matrix $A_{2,1}$ for different values of the iteration number $q$.
The standard MKL SVD routine errors are shown as a thick black line. }
\label{fig:figure1}
\end{figure}
\begin{figure}[h]
\subfigure[]{\label{fig:figure2a} \includegraphics[width=\columnwidth]{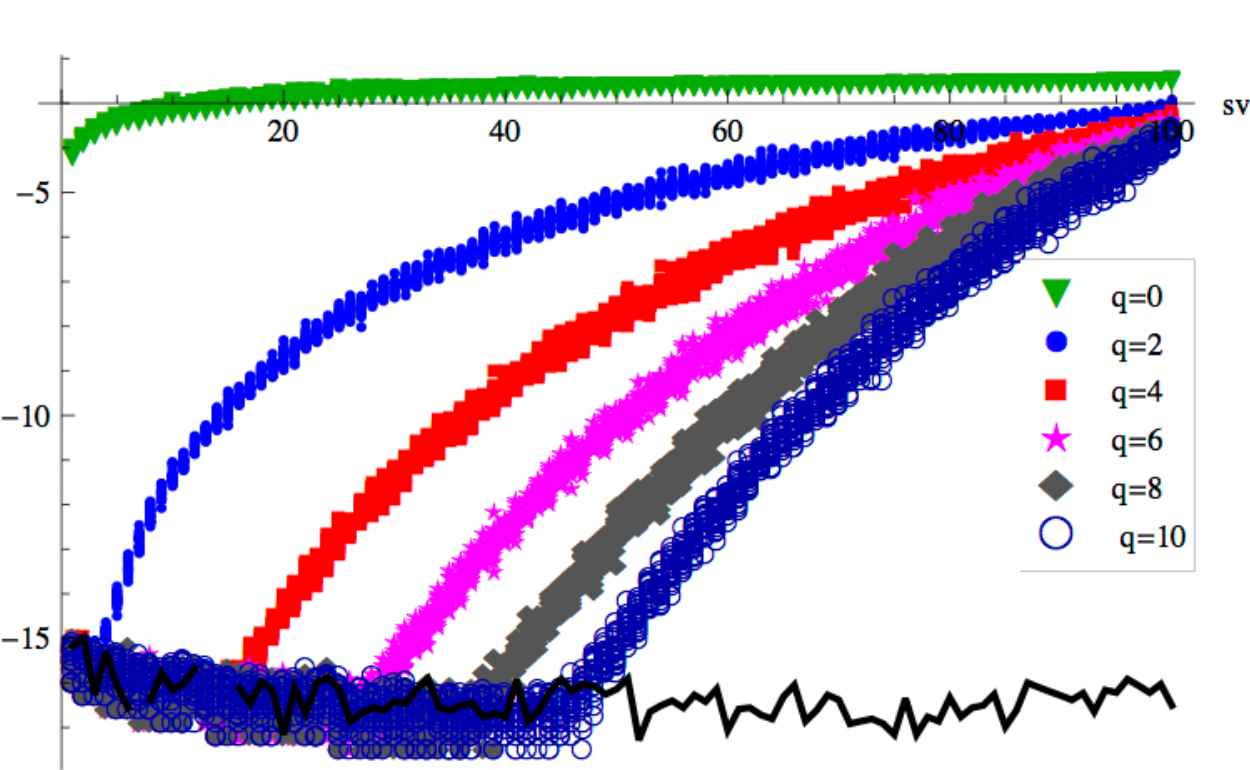}}
%\subfigure[]{\label{fig:figure2b} \includegraphics[width=0.45 \columnwidth]{errCritical.pdf}}
\caption{Instance $A_c$: $m=1500,n=750$; $k=p=50$. The errors $\log_{10}(|\sigma_i - \sigma_i^{RRSVD}|)$ 
of the RRSVD  as in \fref{fig:figure1b} but referring to the singular values of the matrix $A_c$ and for 
$q \in \{0,2,4,6,8,10\}$. The discarded weight $w$ for the chosen value of $k$ is $w=1 \cdot 10^{-1}$. }
\label{fig:figure1}
\end{figure}
%
%\\In order to assess the settings and range of parameters in which RRSVD can successfully improve TEBD,
We now  turn our attention to the TEBD extremal case discussed previously.
We consider an $m=1500, n=750$ random matrix $A_c$ generated, as described at the beginning of this subsection, starting from singular values $\Sigma_c =diag(\sigma_1^c,\sigma_2^c,\ldots,\sigma_n^c) $ with $\sigma_i^c  =1/i$.
As shown in \fref{fig:figure2a}, in order to provide the same accuracy delivered by LAPack on the first $k$ singular values, we need to increase the number of iterations.
For $k=50$ and $q=10$ RRSVD is still able to provide some speed-up over the LAPack SVD.
But the real problem is the truncation error: for $k=50$ we have a truncation error of $\epsilon\approx10^{-1}$. But in order to achieve $\epsilon<10^{-2}$, about $650$ singular values need to be retained. This results in a major loss of efficiency of the TEBD simulation scheme.

Therefore we can claim that RRSVD is indeed a fast and reliable method, able to successfully replace the 
standard SVD in  the TEBD algorithm \emph{in all situations} where TEBD can successfully be applied.
\subsection{Performance on larger matrices and TEBD profiling} \label{sec:profiling}
Now that the basic properties of the algorithm are established, we test it on larger 
matrices sampled from mixed-state TEBD simulation of our benchmark system from section \ref{sec:tedopa}.
Given the bond dimension $\chi$ and the local dimension $d$ of the sites involved in 
the two-site-update, the size of the matrix given as input to the SVD is  $d^2 \chi \times d^2 \chi$.
In the following example we set $k=\chi=100$  and the dimension of the local oscillators 
to $3,4,5,6$ and $7$ respectively. We therefore present results for matrices of dimensions 
$d_3=900 \times 900, d_4= 1600 \times1600, d_5 =2500 \times 2500$, $d_6 = 3600 \times 3600$ 
and $d_7  = 4900 \times 4900$. The structural properties of the test matrices considered 
are similar to the non-critical instances considered in the previous subsection.

We first analyzed the results provided by the RRSVD routine on a large sample of matrices (200 instances for each dimension).
We determined that the RRSVD reaches the LAPack accuracy for a number $q=2$ of PI steps.

We have developed three versions of the RRSVD algorithm: BL-, MKL- and GPU-RRSVD; each one 
uses a different implementation of the BLAS and LAPack libraries. BL-RRSVD  is based on a 
standard single-thread implementation  (CBLAS \cite{cblas}, LAPACKE \cite{lapacke}); MKL-RRSVD 
uses the Intel$^ \circledR$  implementation  Math Kernel Library (MKL \cite{mkl}); GPU-RRSVD  
exploits CUBLAS\cite{cublas} and CULA \cite{cula}, i.e. the BLAS and LAPack libraries for 
Nvidia$^\circledR$ Graphics Processing Units (GPUs). RRSVD is available for single/double 
precision real/complex matrices in each version. We refer the reader to  \cite{RRSVDgit} 
for more details about our  RRSVD implementations.
During the completion of this work,  another implementation of RRSVD from one of the 
authors of \cite{halko11} has been reported in \cite{martinsson15}. There the authors 
present  three variants  of RRSVD (essentially RRSVD with and without the PI and the 
Accuracy check) and discuss their performance on large (up to $6000 \times 12000$) real 
matrices with slowly decaying singular values. The implementation described in \cite{martinsson15} 
is available for single-multi and Kilo processor architectures, as ours, but is limited 
to double precision real matrices.  We are currently working on a  full comparison between 
our and this other version of RRSVD.

In Table \ref{tab:baretimes} we show the time required to perform the SVD/RRSVD of $
d_3,d_4,d_5$,$d_6$ and $d_7$ double-precision complex matrices for the three implementations.

\begin{table*}[t]
\begin{center}
\begin{tabular}{||c|| | c |c | c | c | c | c | c | c ||}
\hline \hline
 & 1-BL-SVD& 1-BL-RRSVD & 1-MKL-SVD & 1-MKL-RRSVD & 16-MKL-SVD &16-MKL-RRSVD  & GPU-SVD & GPU-RRSVD\\ \hline \hline
 $d_3$ &  13.27 & 6.62 & 1.47 &0.71 & 0.46 & 0.14 & 1.08 & 0.25 \\ \hline
 $d_4$ &  262.47 &  21.38 & 9.31 & 1.69 & 1.92 & 0.36 & 3.92 & 0.41 \\ \hline
 $d_5$ &  449.97 & 37.19 & 31.87 & 3.62 & 6.07 & 0.54 & 9.95 & 0.61 \\ \hline
$d_6$ &  1464.67 & 74.48 & 97.37 &6.97& 22.93 & 0.84  & 21.97 &0.88 \\ \hline
$d_7$ &  1973.23& 99.9 & 241.01 & 11.49 & 61.51 & 1.43  & 49.00 &1.48 \\ \hline\hline
\end{tabular}
\end{center}
\caption{Execution time, in seconds, for the SVD of matrices of different sizes with LAPack SVD and RRSVD.
The parameters for RRSVD are $q=2$, $k=p=100$.
1-BL: single-thread  CBLAS-LAPACKE-based implementation, executed on a Intel i7@2.66GHz processor.
MKL: MKL-based implementation; 1-MKL: with one MKL thread,  16-MKL: with 16 MKL-threads, executed on one and  two  8-core Xeon X5570@2.93GHz respectively.
GPU: CUBLAS-CULA implementation executed on a NVIDIA K20s; the timing in this case includes host-to-device and device-to-host memory transfers.}
\label{tab:baretimes}
\end{table*}%
RRSVD provides a speed-up (\fref{fig:speedup}) which is consistent with the predictions except in the 16-MKL case, there
it grows stronger than linearly for matrix sizes in the range $d_3-d_6$. This peculiar behavior is due to the low impact
of the initial matrix multiplication $Y = A \Omega$ (and subsequent ones) on matrices of such sizes: this operation is
heavily optimized for multi-threaded execution. Then the operations that determine the computational cost are the QR and
final SVD decompositions,  which have complexity $\mathcal{O}\llrr{m \cdot(k+p)^2}$. Since $k+p$ is kept constant, we have
a quadratic speed-up. This justification is supported, for example, by the speed-up scaling of the 1-MKL case.
However, as the matrix size increases the speed-up will tend to be linear (see the $d_7$ case).
The performance  on RRSVD provided by GPUs are comparable those delivered by the 16-MKL .
Indeed, the standard SVD is faster on the GPU starting from size $d_6$.
This behavior was expected, since our test matrices are still too small to take full advantage of the Kilo-processor architecture.
%
%\subsection{Full TEBD simulations profiling}

At last, we  perform full TEBD simulations:  for each dimension $d_i, \ i=3,4,5,6,7$ we executed one 
TEBD simulation with  the standard SVD routine, and another TEBD with RRSVD. We  ran all the jobs on 
the same  cluster node, equipped with two  Xeon X5570@2.93GHz with 8 cores each, as to assure a fair 
comparison.
In \fref{fig:profile} we show the average time required by the TEBD two-site update when standard 
SVD and RRSVD are used. The overall speed-up of the two-site update (inset of \fref{fig:profile}) 
agrees with the Amdahl law \cite{amdahl67}. The average is taken over all maximum-sized two-site 
updates performed in the TEBD simulation (more than 1500 for every dimension). When the standard 
SVD is applied, it takes more than 90\% of the overall computation time; however if RRSVD is employed, 
SVD-times reduce drastically. Its computational cost is of the same order as that of the building 
of the $\widetilde{\Theta}$ matrix (cf. \fref{fig:profile}). The fluctuations around the average RRSVD
time are due to the action of the Accuracy-Check: from time to time the bond dimension $\chi$ must 
be increased in order to keep the approximation error below the threshold value $\epsilon = 10^{-3}$ 
required in the simulation. According to  Corollary \ref{th:tamaErr}, for the choice $p=100$ and $ q=2$, 
such an increase is motivated only by the decay rate of the singular values: the failure probability 
is smaller than $1/10^{100}$. This is confirmed by an  \emph{a posteriori} analysis of the matrices 
produced during our test TEBD simulations that required an extension of the bond dimension. The 
rank of the approximation proposed by RRSVD, for  assigned tolerance $\epsilon$, can therefore be 
used to detect either an entanglement build-up in the system or  an (unlikely) anomalous behavior of the RRSVD itself.
%
%%
%\\[5pt]While the speed-up of the decimation step of TEBD can therefore be considerable, this does not mean that we can speed-up the whole TEBD update by the same amount (inset of figure ???). This ``fact of life'' is formalized by the Amdahl law \cite{amdahl67}.

%
%\\[5pt]
%
%\begin{enumerate}
%\item Introduce the example (Monomer, with different local oscillator sizes): why different sizes. Behavior of the MPO entanglement?
%\item Show the results on accuracy or explain that they are equal to what shown before
%\item Show the timing-results.
%\item Observe that now the construction of the ThetaTilde matrix is the bottleneck and that this routine can be further improved.
%\end{enumerate}
%
%%
\begin{figure}[t]
\includegraphics[width=\columnwidth]{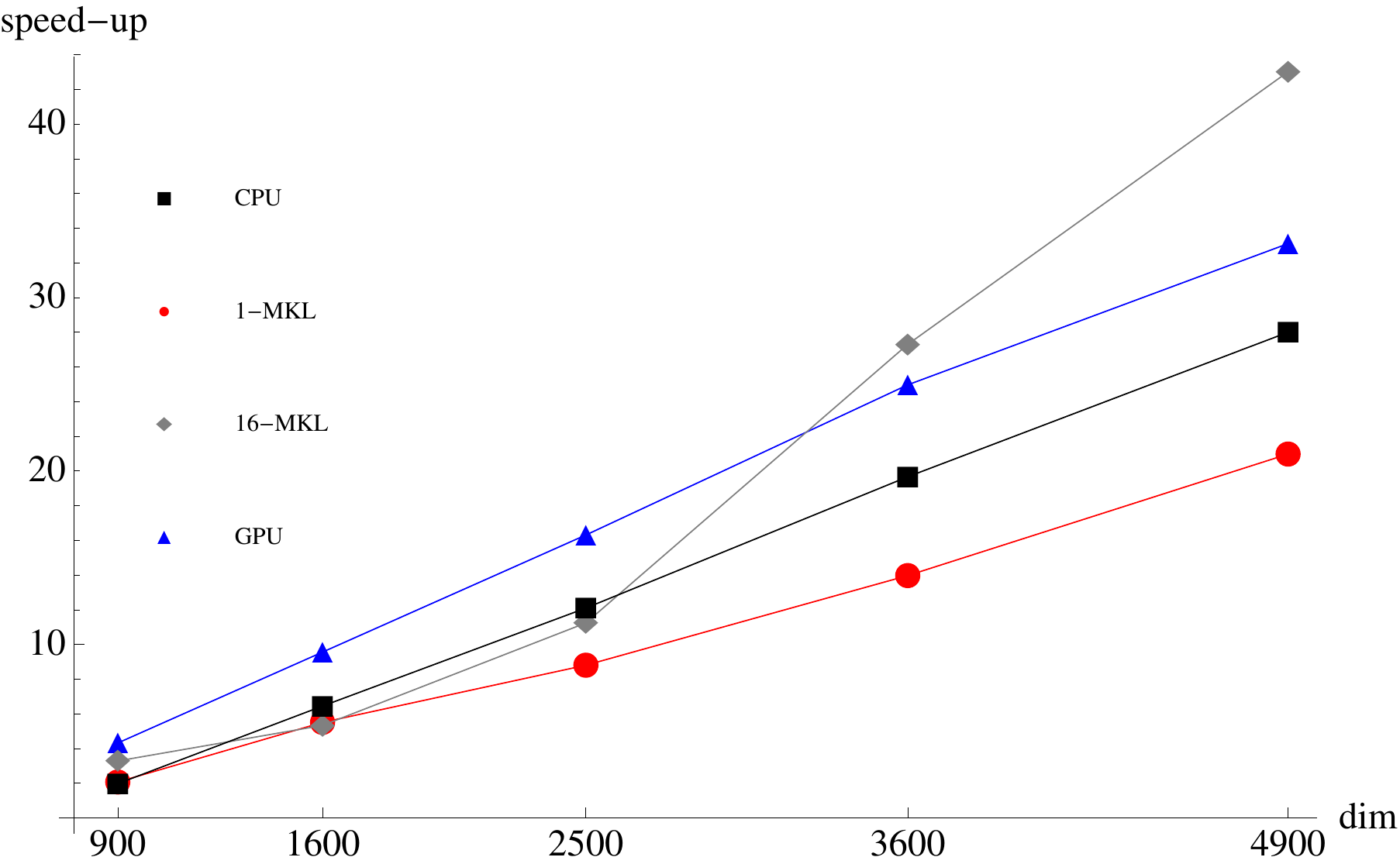}
\caption{$k=p=100$; $q=2$. Scaling of the SVD/RRSVD speed-up for different platforms as a function of the matrix size, as obtained from the data reported in Table \ref{tab:baretimes}.}
\label{fig:speedup}
\end{figure}
\begin{figure}[t]
\includegraphics[width=\columnwidth]{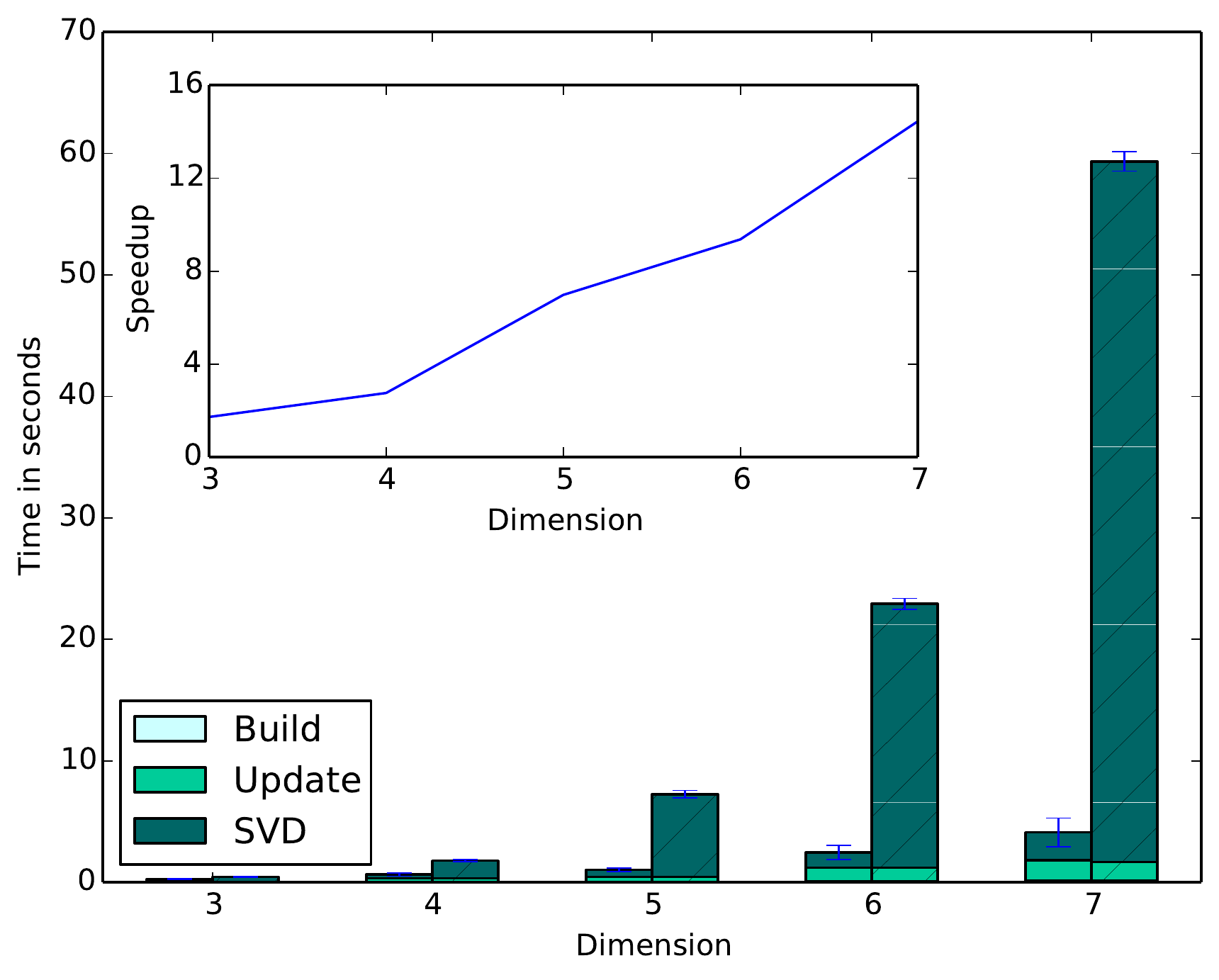}
\caption{\label{fig:profile} Profile of the average time required by the SVD during a TEBD update step; left bars are for the RRSVD,
right bars for the standard SVD. The data refers to the MKL implementation executed with $16$ MKL-threads. The (unbiased) standard deviation
of the execution time for the singular value decomposition part of the update step is shown as an error bar on top of each bar.
The inset presents the overall speed-up for one complete TEBD two-site update achieved by the RRSVD.}
\end{figure}

\section{Conclusions and Outlook}
The TEBD-algorithm, an important tool for the numerical description of one-dimensional quantum many-body
systems, depends essentially on the SVD which, in current implementations represents the bottleneck. We
have demonstrated that this bottleneck can be addressed by the successful application of the RRSVD algorithm
to TEBD simulations. The block decimation step of the TEBD update procedure is now approximately one order
of magnitude faster than with the standard SVD without incurring additional accuracy losses. We note that
in our test cases we have always chosen the RRSVD parameters such that we obtain singular values (and vectors)
which were as precise as those provided by the standard (LAPack) SVD routine. By relaxing this requirement,
the speed-up can be increased further. Moreover, by augmenting RRSVD with the Accuracy Check feature we are
able not only to detect any very unlikely deviations from the average behavior, but also to understand whether
the system is experiencing an entanglement build-up which would require an increasing of the retained Hilbert space.

In this paper we focused on the TEBD algorithm and its application to the one-dimensional system obtained
through a TEDOPA mapping of an open quantum system (section \ref{sec:tedopa}). In this context, RRSVD makes
it possible to increase the dimension of the local oscillators with a much reduced impact on the computational
time, thus allowing for the efficient simulation of the system at higher temperatures. However, all MPS algorithms
relying on the SVD to select the relevant Hilbert space can greatly benefit from the use of the RRSVD routine,
as long as the ratio between the number of retained and total singular values is sufficiently small.

The real scope and impact of this new computational tool is still to be understood fully. To this end, we
prepared the \emph{RRSVD-Package}, a library that provides the RRSVD routine for single-/double-precision
and real-/complex-matrices. The package includes all the different implementations (BL, MKL, GPU) of RRSVD
and has been designed to be plugged into existing codes through very minor code modifications: in principle,
it suffices to replace each call to SVD by a call to RRSVD. This should allow for a very quick test of RRSVD
in different simulation codes and scenarios. The library is written in C++; a Fortran wrapper, that allows
to call of the RRSVD routine from Fortran code, is included as well. Some of the matrices used for the analysis
of RRSVD in this paper are available, together with some stand-alone code that exemplifies the use of RRSVD
and how to reproduce some of the results reported in this paper. The RRSVD-Package is freely available at
\cite{RRSVDgit} .
%
%
%\\[5pt]
%TENTATIVE: GPU TEBD: is someone interested???
\\[5pt]
The results obtained for the GPU implementation are rather promising: for the largest matrices considered
($d_5,d_6,d_7$) the GPU performs as well as the cluster node. Preliminary results on even larger matrices
show that a GPU can become a valid alternative means to perform TEBD simulations of system with high local
dimensions, or when the number of retained dimensions must be increased because of larger correlation lengths.
Moreover, if operated in the right way, a GPU can act as a large cluster of machines \cite{tama14} without
the difficulties stemming from the need of distributing the workload among different computational nodes
(Message Passing Interface (MPI)). A full GPU version of TEBD can make the access to super-computing facilities
superfluous: a typical laboratory workstation equipped with one or two GPUs would be sufficient. We are
currently re-analyzing the TEBD algorithm to expose further options for parallelization, as for example in
the construction of the $\widetilde{\Theta}$ matrix. It could be computed by an ad-hoc designed CUDA-kernel
and is a valid target for improvement now that its computational complexity is similar to that of the SVD.
\section*{Acknowledgements}
This work was supported by an Alexander von Humboldt-Professorship, the EU Integrating project SIQS, the
EU STREP projects PAPETS and EQUAM, and the ERC Synergy grant BioQ. The simulations were performed on the
computational resource bwUniCluster funded by the Ministry of Science, Research and Arts and the Universities
of the State of Baden-W\"urttemberg, Germany, within the framework program bwHPC.


\begin{thebibliography}{50}%
\makeatletter
\providecommand \@ifxundefined [1]{%
 \@ifx{#1\undefined}
}%
\providecommand \@ifnum [1]{%
 \ifnum #1\expandafter \@firstoftwo
 \else \expandafter \@secondoftwo
 \fi
}%
\providecommand \@ifx [1]{%
 \ifx #1\expandafter \@firstoftwo
 \else \expandafter \@secondoftwo
 \fi
}%
\providecommand \natexlab [1]{#1}%
\providecommand \enquote  [1]{``#1''}%
\providecommand \bibnamefont  [1]{#1}%
\providecommand \bibfnamefont [1]{#1}%
\providecommand \citenamefont [1]{#1}%
\providecommand \href@noop [0]{\@secondoftwo}%
\providecommand \href [0]{\begingroup \@sanitize@url \@href}%
\providecommand \@href[1]{\@@startlink{#1}\@@href}%
\providecommand \@@href[1]{\endgroup#1\@@endlink}%
\providecommand \@sanitize@url [0]{\catcode `\\12\catcode `\$12\catcode
  `\&12\catcode `\#12\catcode `\^12\catcode `\_12\catcode `\%12\relax}%
\providecommand \@@startlink[1]{}%
\providecommand \@@endlink[0]{}%
\providecommand \url  [0]{\begingroup\@sanitize@url \@url }%
\providecommand \@url [1]{\endgroup\@href {#1}{\urlprefix }}%
\providecommand \urlprefix  [0]{URL }%
\providecommand \Eprint [0]{\href }%
\providecommand \doibase [0]{http://dx.doi.org/}%
\providecommand \selectlanguage [0]{\@gobble}%
\providecommand \bibinfo  [0]{\@secondoftwo}%
\providecommand \bibfield  [0]{\@secondoftwo}%
\providecommand \translation [1]{[#1]}%
\providecommand \BibitemOpen [0]{}%
\providecommand \bibitemStop [0]{}%
\providecommand \bibitemNoStop [0]{.\EOS\space}%
\providecommand \EOS [0]{\spacefactor3000\relax}%
\providecommand \BibitemShut  [1]{\csname bibitem#1\endcsname}%
\let\auto@bib@innerbib\@empty
%</preamble>
\bibitem [{\citenamefont {Jozsa}\ and\ \citenamefont {Linden}(2003)}]{jozsa03}%
  \BibitemOpen
  \bibfield  {author} {\bibinfo {author} {\bibfnamefont {R.}~\bibnamefont
  {Jozsa}}\ and\ \bibinfo {author} {\bibfnamefont {N.}~\bibnamefont {Linden}},\
  }\href@noop {} {\bibfield  {journal} {\bibinfo  {journal} {Proc. R. Soc.
  Lond. A}\ }\textbf {\bibinfo {volume} {459}},\ \bibinfo {pages} {2011}
  (\bibinfo {year} {2003})}\BibitemShut {NoStop}%
\bibitem [{\citenamefont {Audenaert}\ \emph {et~al.}(2002)\citenamefont
  {Audenaert}, \citenamefont {Eisert}, \citenamefont {Plenio},\ and\
  \citenamefont {Werner}}]{Audenaert2002}%
  \BibitemOpen
  \bibfield  {author} {\bibinfo {author} {\bibfnamefont {K.}~\bibnamefont
  {Audenaert}}, \bibinfo {author} {\bibfnamefont {J.}~\bibnamefont {Eisert}},
  \bibinfo {author} {\bibfnamefont {M.~B.}\ \bibnamefont {Plenio}}, \ and\
  \bibinfo {author} {\bibfnamefont {R.~F.}\ \bibnamefont {Werner}},\
  }\href@noop {} {\bibfield  {journal} {\bibinfo  {journal} {Phys. Rev. A}\
  }\textbf {\bibinfo {volume} {66}},\ \bibinfo {pages} {042327} (\bibinfo
  {year} {2002})}\BibitemShut {NoStop}%
\bibitem [{\citenamefont {Plenio}\ \emph {et~al.}(2005)\citenamefont {Plenio},
  \citenamefont {Eisert}, \citenamefont {Drei{\ss}ig},\ and\ \citenamefont
  {Cramer}}]{Plenio2005}%
  \BibitemOpen
  \bibfield  {author} {\bibinfo {author} {\bibfnamefont {M.~B.}\ \bibnamefont
  {Plenio}}, \bibinfo {author} {\bibfnamefont {J.}~\bibnamefont {Eisert}},
  \bibinfo {author} {\bibfnamefont {J.}~\bibnamefont {Drei{\ss}ig}}, \ and\
  \bibinfo {author} {\bibfnamefont {M.}~\bibnamefont {Cramer}},\ }\href@noop {}
  {\bibfield  {journal} {\bibinfo  {journal} {Phys. Rev. Lett.}\ }\textbf
  {\bibinfo {volume} {94}},\ \bibinfo {pages} {060503} (\bibinfo {year}
  {2005})}\BibitemShut {NoStop}%
\bibitem [{\citenamefont {Eisert}\ \emph {et~al.}(2010)\citenamefont {Eisert},
  \citenamefont {Cramer},\ and\ \citenamefont {Plenio}}]{eisert08}%
  \BibitemOpen
  \bibfield  {author} {\bibinfo {author} {\bibfnamefont {J.}~\bibnamefont
  {Eisert}}, \bibinfo {author} {\bibfnamefont {M.}~\bibnamefont {Cramer}}, \
  and\ \bibinfo {author} {\bibfnamefont {M.~B.}\ \bibnamefont {Plenio}},\
  }\href@noop {} {\bibfield  {journal} {\bibinfo  {journal} {Rev. Mod. Phys.}\
  }\textbf {\bibinfo {volume} {82}},\ \bibinfo {pages} {277} (\bibinfo {year}
  {2010})}\BibitemShut {NoStop}%
\bibitem [{\citenamefont {Brand{\~a}o}\ and\ \citenamefont
  {Horodecki}(2013)}]{horodecki13}%
  \BibitemOpen
  \bibfield  {author} {\bibinfo {author} {\bibfnamefont {F.~G.~S.~L.}\
  \bibnamefont {Brand{\~a}o}}\ and\ \bibinfo {author} {\bibfnamefont
  {M.}~\bibnamefont {Horodecki}},\ }\href@noop {} {\bibfield  {journal}
  {\bibinfo  {journal} {Nat. Phys.}\ }\textbf {\bibinfo {volume} {9}},\
  \bibinfo {pages} {721} (\bibinfo {year} {2013})}\BibitemShut {NoStop}%
\bibitem [{\citenamefont {Brand{\~a}o}\ and\ \citenamefont
  {Horodecki}(2015)}]{horodecki15}%
  \BibitemOpen
  \bibfield  {author} {\bibinfo {author} {\bibfnamefont {F.~G.~S.~L.}\
  \bibnamefont {Brand{\~a}o}}\ and\ \bibinfo {author} {\bibfnamefont
  {M.}~\bibnamefont {Horodecki}},\ }\href@noop {} {\bibfield  {journal}
  {\bibinfo  {journal} {Comm. Math. Phys.}\ }\textbf {\bibinfo {volume}
  {333}},\ \bibinfo {pages} {761} (\bibinfo {year} {2015})}\BibitemShut
  {NoStop}%
\bibitem [{\citenamefont {White}(1992)}]{white92}%
  \BibitemOpen
  \bibfield  {author} {\bibinfo {author} {\bibfnamefont {S.~R.}\ \bibnamefont
  {White}},\ }\href@noop {} {\bibfield  {journal} {\bibinfo  {journal} {Phys.
  Rev. Lett.}\ }\textbf {\bibinfo {volume} {69}},\ \bibinfo {pages} {2863}
  (\bibinfo {year} {1992})}\BibitemShut {NoStop}%
\bibitem [{\citenamefont {Rommer}\ and\ \citenamefont
  {\"Ostlund}(1997)}]{ref:rommer1997}%
  \BibitemOpen
  \bibfield  {author} {\bibinfo {author} {\bibfnamefont {S.}~\bibnamefont
  {Rommer}}\ and\ \bibinfo {author} {\bibfnamefont {S.}~\bibnamefont
  {\"Ostlund}},\ }\href {\doibase 10.1103/PhysRevB.55.2164} {\bibfield
  {journal} {\bibinfo  {journal} {Phys. Rev. B}\ }\textbf {\bibinfo {volume}
  {55}},\ \bibinfo {pages} {2164} (\bibinfo {year} {1997})}\BibitemShut
  {NoStop}%
\bibitem [{\citenamefont {Perez-Garcia}\ \emph {et~al.}(2007)\citenamefont
  {Perez-Garcia}, \citenamefont {Verstraete}, \citenamefont {Wolf},\ and\
  \citenamefont {Cirac}}]{ref:perezgarcia2007}%
  \BibitemOpen
  \bibfield  {author} {\bibinfo {author} {\bibfnamefont {D.}~\bibnamefont
  {Perez-Garcia}}, \bibinfo {author} {\bibfnamefont {F.}~\bibnamefont
  {Verstraete}}, \bibinfo {author} {\bibfnamefont {M.~M.}\ \bibnamefont
  {Wolf}}, \ and\ \bibinfo {author} {\bibfnamefont {J.~I.}\ \bibnamefont
  {Cirac}},\ }\href {http://dl.acm.org/citation.cfm?id=2011832.2011833}
  {\bibfield  {journal} {\bibinfo  {journal} {Quantum Info. Comput.}\ }\textbf
  {\bibinfo {volume} {7}},\ \bibinfo {pages} {401} (\bibinfo {year}
  {2007})}\BibitemShut {NoStop}%
\bibitem [{\citenamefont {White}\ and\ \citenamefont
  {Feiguin}(2004)}]{white04}%
  \BibitemOpen
  \bibfield  {author} {\bibinfo {author} {\bibfnamefont {S.~R.}\ \bibnamefont
  {White}}\ and\ \bibinfo {author} {\bibfnamefont {A.~E.}\ \bibnamefont
  {Feiguin}},\ }\href@noop {} {\bibfield  {journal} {\bibinfo  {journal} {Phys.
  Rev. Lett.}\ }\textbf {\bibinfo {volume} {93}},\ \bibinfo {pages} {076401}
  (\bibinfo {year} {2004})}\BibitemShut {NoStop}%
\bibitem [{\citenamefont {Garc\`ia-Ripoll}(2006)}]{garcia06}%
  \BibitemOpen
  \bibfield  {author} {\bibinfo {author} {\bibfnamefont {J.~J.}\ \bibnamefont
  {Garc\`ia-Ripoll}},\ }\href@noop {} {\bibfield  {journal} {\bibinfo
  {journal} {New J. Phys.}\ }\textbf {\bibinfo {volume} {8}},\ \bibinfo {pages}
  {305} (\bibinfo {year} {2006})}\BibitemShut {NoStop}%
\bibitem [{\citenamefont {Vidal}(2003)}]{ref:vidal2003}%
  \BibitemOpen
  \bibfield  {author} {\bibinfo {author} {\bibfnamefont {G.}~\bibnamefont
  {Vidal}},\ }\href {\doibase 10.1103/PhysRevLett.91.147902} {\bibfield
  {journal} {\bibinfo  {journal} {Phys. Rev. Lett.}\ }\textbf {\bibinfo
  {volume} {91}},\ \bibinfo {pages} {147902} (\bibinfo {year}
  {2003})}\BibitemShut {NoStop}%
\bibitem [{\citenamefont {Vidal}(2004)}]{ref:vidal2004}%
  \BibitemOpen
  \bibfield  {author} {\bibinfo {author} {\bibfnamefont {G.}~\bibnamefont
  {Vidal}},\ }\href {\doibase 10.1103/PhysRevLett.93.040502} {\bibfield
  {journal} {\bibinfo  {journal} {Phys. Rev. Lett.}\ }\textbf {\bibinfo
  {volume} {93}},\ \bibinfo {pages} {040502} (\bibinfo {year}
  {2004})}\BibitemShut {NoStop}%
\bibitem [{\citenamefont {Schollw{\"o}ck}(2011)}]{ref:schollwoeck2011}%
  \BibitemOpen
  \bibfield  {author} {\bibinfo {author} {\bibfnamefont {U.}~\bibnamefont
  {Schollw{\"o}ck}},\ }\href@noop {} {\bibfield  {journal} {\bibinfo  {journal}
  {Ann. Phys.}\ }\textbf {\bibinfo {volume} {326}},\ \bibinfo {pages} {96}
  (\bibinfo {year} {2011})}\BibitemShut {NoStop}%
\bibitem [{\citenamefont {Halko}\ \emph {et~al.}(2011)\citenamefont {Halko},
  \citenamefont {Martinsson},\ and\ \citenamefont {Tropp}}]{halko11}%
  \BibitemOpen
  \bibfield  {author} {\bibinfo {author} {\bibfnamefont {N.}~\bibnamefont
  {Halko}}, \bibinfo {author} {\bibfnamefont {P.}~\bibnamefont {Martinsson}}, \
  and\ \bibinfo {author} {\bibfnamefont {J.}~\bibnamefont {Tropp}},\
  }\href@noop {} {\bibfield  {journal} {\bibinfo  {journal} {SIAM Review}\
  }\textbf {\bibinfo {volume} {53}},\ \bibinfo {pages} {217} (\bibinfo {year}
  {2011})}\BibitemShut {NoStop}%
\bibitem [{\citenamefont {Zwolak}\ and\ \citenamefont
  {Vidal}(2004)}]{ref:zwolak2004}%
  \BibitemOpen
  \bibfield  {author} {\bibinfo {author} {\bibfnamefont {M.}~\bibnamefont
  {Zwolak}}\ and\ \bibinfo {author} {\bibfnamefont {G.}~\bibnamefont {Vidal}},\
  }\href {\doibase 10.1103/PhysRevLett.93.207205} {\bibfield  {journal}
  {\bibinfo  {journal} {Phys. Rev. Lett.}\ }\textbf {\bibinfo {volume} {93}},\
  \bibinfo {pages} {207205} (\bibinfo {year} {2004})}\BibitemShut {NoStop}%
\bibitem [{\citenamefont {Werner}\ \emph {et~al.}(2014)\citenamefont {Werner},
  \citenamefont {Jaschke}, \citenamefont {Silvi}, \citenamefont {Calarco},
  \citenamefont {Eisert},\ and\ \citenamefont {Montangero}}]{Montangero14}%
  \BibitemOpen
  \bibfield  {author} {\bibinfo {author} {\bibfnamefont {A.~H.}\ \bibnamefont
  {Werner}}, \bibinfo {author} {\bibfnamefont {D.}~\bibnamefont {Jaschke}},
  \bibinfo {author} {\bibfnamefont {P.}~\bibnamefont {Silvi}}, \bibinfo
  {author} {\bibfnamefont {T.}~\bibnamefont {Calarco}}, \bibinfo {author}
  {\bibfnamefont {J.}~\bibnamefont {Eisert}}, \ and\ \bibinfo {author}
  {\bibfnamefont {S.}~\bibnamefont {Montangero}},\ }\href@noop {} {\bibfield
  {journal} {\bibinfo  {journal} {E-print arXiv:1412.5746}\ } (\bibinfo {year}
  {2014})}\BibitemShut {NoStop}%
\bibitem [{\citenamefont {Suzuki}(1990)}]{ref:suzuki1990}%
  \BibitemOpen
  \bibfield  {author} {\bibinfo {author} {\bibfnamefont {M.}~\bibnamefont
  {Suzuki}},\ }\href {\doibase 10.1016/0375-9601(90)90962-N} {\bibfield
  {journal} {\bibinfo  {journal} {Phys. Lett. A}\ }\textbf {\bibinfo {volume}
  {146}},\ \bibinfo {pages} {319} (\bibinfo {year} {1990})}\BibitemShut
  {NoStop}%
\bibitem [{\citenamefont {Hatano}\ and\ \citenamefont
  {Suzuki}(2005)}]{ref:suzuki2005}%
  \BibitemOpen
  \bibfield  {author} {\bibinfo {author} {\bibfnamefont {N.}~\bibnamefont
  {Hatano}}\ and\ \bibinfo {author} {\bibfnamefont {M.}~\bibnamefont
  {Suzuki}},\ }in\ \href {\doibase 10.1007/11526216_2} {\emph {\bibinfo
  {booktitle} {Quantum Annealing and Other Optimization Methods}}},\ \bibinfo
  {series} {Lecture Notes in Physics}, Vol.\ \bibinfo {volume} {679},\ \bibinfo
  {editor} {edited by\ \bibinfo {editor} {\bibfnamefont {A.}~\bibnamefont
  {Das}}\ and\ \bibinfo {editor} {\bibfnamefont {B.}~\bibnamefont
  {K.~Chakrabarti}}}\ (\bibinfo  {publisher} {Springer Berlin Heidelberg},\
  \bibinfo {year} {2005})\ pp.\ \bibinfo {pages} {37--68}\BibitemShut {NoStop}%
\bibitem [{\citenamefont {Sornborger}\ and\ \citenamefont
  {Stewart}(1999)}]{ref:sornborger1999}%
  \BibitemOpen
  \bibfield  {author} {\bibinfo {author} {\bibfnamefont {A.~T.}\ \bibnamefont
  {Sornborger}}\ and\ \bibinfo {author} {\bibfnamefont {E.~D.}\ \bibnamefont
  {Stewart}},\ }\href {\doibase 10.1103/PhysRevA.60.1956} {\bibfield  {journal}
  {\bibinfo  {journal} {Phys. Rev. A}\ }\textbf {\bibinfo {volume} {60}},\
  \bibinfo {pages} {1956} (\bibinfo {year} {1999})}\BibitemShut {NoStop}%
\bibitem [{\citenamefont {Gobert}\ \emph {et~al.}(2005)\citenamefont {Gobert},
  \citenamefont {Kollath}, \citenamefont {Schollw\"ock},\ and\ \citenamefont
  {Sch\"utz}}]{ref:gobert2005}%
  \BibitemOpen
  \bibfield  {author} {\bibinfo {author} {\bibfnamefont {D.}~\bibnamefont
  {Gobert}}, \bibinfo {author} {\bibfnamefont {C.}~\bibnamefont {Kollath}},
  \bibinfo {author} {\bibfnamefont {U.}~\bibnamefont {Schollw\"ock}}, \ and\
  \bibinfo {author} {\bibfnamefont {G.}~\bibnamefont {Sch\"utz}},\ }\href
  {\doibase 10.1103/PhysRevE.71.036102} {\bibfield  {journal} {\bibinfo
  {journal} {Phys. Rev. E}\ }\textbf {\bibinfo {volume} {71}},\ \bibinfo
  {pages} {036102} (\bibinfo {year} {2005})}\BibitemShut {NoStop}%
\bibitem [{\citenamefont {Prior}\ \emph {et~al.}(2010)\citenamefont {Prior},
  \citenamefont {Chin}, \citenamefont {Huelga},\ and\ \citenamefont
  {Plenio}}]{ref:prior2010}%
  \BibitemOpen
  \bibfield  {author} {\bibinfo {author} {\bibfnamefont {J.}~\bibnamefont
  {Prior}}, \bibinfo {author} {\bibfnamefont {A.~W.}\ \bibnamefont {Chin}},
  \bibinfo {author} {\bibfnamefont {S.~F.}\ \bibnamefont {Huelga}}, \ and\
  \bibinfo {author} {\bibfnamefont {M.~B.}\ \bibnamefont {Plenio}},\ }\href
  {\doibase 10.1103/PhysRevLett.105.050404} {\bibfield  {journal} {\bibinfo
  {journal} {Phys. Rev. Lett.}\ }\textbf {\bibinfo {volume} {105}},\ \bibinfo
  {pages} {050404} (\bibinfo {year} {2010})}\BibitemShut {NoStop}%
\bibitem [{\citenamefont {Chin}\ \emph {et~al.}(2010)\citenamefont {Chin},
  \citenamefont {Rivas}, \citenamefont {Huelga},\ and\ \citenamefont
  {Plenio}}]{ref:chin2010}%
  \BibitemOpen
  \bibfield  {author} {\bibinfo {author} {\bibfnamefont {A.~W.}\ \bibnamefont
  {Chin}}, \bibinfo {author} {\bibfnamefont {A.}~\bibnamefont {Rivas}},
  \bibinfo {author} {\bibfnamefont {S.~F.}\ \bibnamefont {Huelga}}, \ and\
  \bibinfo {author} {\bibfnamefont {M.~B.}\ \bibnamefont {Plenio}},\
  }\href@noop {} {\bibfield  {journal} {\bibinfo  {journal} {J. of Math.
  Phys.}\ }\textbf {\bibinfo {volume} {51}},\ \bibinfo {pages} {092109}
  (\bibinfo {year} {2010})}\BibitemShut {NoStop}%
\bibitem [{\citenamefont {Woods}\ \emph {et~al.}(2014)\citenamefont {Woods},
  \citenamefont {Groux}, \citenamefont {Chin}, \citenamefont {Huelga},\ and\
  \citenamefont {Plenio}}]{ref:woods2014}%
  \BibitemOpen
  \bibfield  {author} {\bibinfo {author} {\bibfnamefont {M.~P.}\ \bibnamefont
  {Woods}}, \bibinfo {author} {\bibfnamefont {R.}~\bibnamefont {Groux}},
  \bibinfo {author} {\bibfnamefont {A.~W.}\ \bibnamefont {Chin}}, \bibinfo
  {author} {\bibfnamefont {S.~F.}\ \bibnamefont {Huelga}}, \ and\ \bibinfo
  {author} {\bibfnamefont {M.~B.}\ \bibnamefont {Plenio}},\ }\href@noop {}
  {\bibfield  {journal} {\bibinfo  {journal} {J. Math. Phys.}\ }\textbf
  {\bibinfo {volume} {55}} (\bibinfo {year} {2014})}\BibitemShut {NoStop}%
\bibitem [{\citenamefont {Woods}\ \emph {et~al.}(2015)\citenamefont {Woods},
  \citenamefont {Cramer},\ and\ \citenamefont {Plenio}}]{ref:woods2015}%
  \BibitemOpen
  \bibfield  {author} {\bibinfo {author} {\bibfnamefont {M.~P.}\ \bibnamefont
  {Woods}}, \bibinfo {author} {\bibfnamefont {M.}~\bibnamefont {Cramer}}, \
  and\ \bibinfo {author} {\bibfnamefont {M.~B.}\ \bibnamefont {Plenio}},\
  }\href@noop {} {\bibfield  {journal} {\bibinfo  {journal} {E-print
  arXiv:1504.0xxxx}\ } (\bibinfo {year} {2015})}\BibitemShut {NoStop}%
\bibitem [{\citenamefont {Gautschi}(1994)}]{ref:gautschi1994}%
  \BibitemOpen
  \bibfield  {author} {\bibinfo {author} {\bibfnamefont {W.}~\bibnamefont
  {Gautschi}},\ }\href {\doibase http://dx.doi.org/10.1145/174603.174605}
  {\bibfield  {journal} {\bibinfo  {journal} {J. Trans. Math. Softw.}\ }\textbf
  {\bibinfo {volume} {20}},\ \bibinfo {pages} {21} (\bibinfo {year}
  {1994})}\BibitemShut {NoStop}%
\bibitem [{\citenamefont {Rosenbach}\ \emph {et~al.}(2015)\citenamefont
  {Rosenbach}, \citenamefont {Cerrillo}, \citenamefont {Huelga}, \citenamefont
  {Cao},\ and\ \citenamefont {Plenio}}]{Rosenbach15}%
  \BibitemOpen
  \bibfield  {author} {\bibinfo {author} {\bibfnamefont {R.}~\bibnamefont
  {Rosenbach}}, \bibinfo {author} {\bibfnamefont {J.}~\bibnamefont {Cerrillo}},
  \bibinfo {author} {\bibfnamefont {S.~F.}\ \bibnamefont {Huelga}}, \bibinfo
  {author} {\bibfnamefont {J.}~\bibnamefont {Cao}}, \ and\ \bibinfo {author}
  {\bibfnamefont {M.~B.}\ \bibnamefont {Plenio}},\ }\href@noop {} {\bibfield
  {journal} {\bibinfo  {journal} {E-print arXiv:1504.0xxxx}\ } (\bibinfo {year}
  {2015})}\BibitemShut {NoStop}%
\bibitem [{\citenamefont {Golub}\ and\ \citenamefont
  {Van~Loan}(1996)}]{golub96}%
  \BibitemOpen
  \bibfield  {author} {\bibinfo {author} {\bibfnamefont {G.~H.}\ \bibnamefont
  {Golub}}\ and\ \bibinfo {author} {\bibfnamefont {C.~F.}\ \bibnamefont
  {Van~Loan}},\ }\href@noop {} {\emph {\bibinfo {title} {Matrix
  computations}}},\ \bibinfo {edition} {3rd}\ ed.\ (\bibinfo  {publisher}
  {Johns Hopkins University Press},\ \bibinfo {year} {1996})\BibitemShut
  {NoStop}%
\bibitem [{\citenamefont {Hastie}\ \emph {et~al.}(2009)\citenamefont {Hastie},
  \citenamefont {Tibshirani},\ and\ \citenamefont {Friedman}}]{hastie09}%
  \BibitemOpen
  \bibfield  {author} {\bibinfo {author} {\bibfnamefont {T.}~\bibnamefont
  {Hastie}}, \bibinfo {author} {\bibfnamefont {R.}~\bibnamefont {Tibshirani}},
  \ and\ \bibinfo {author} {\bibfnamefont {J.}~\bibnamefont {Friedman}},\
  }\href@noop {} {\emph {\bibinfo {title} {The Elements of Statistical
  Learning}}},\ \bibinfo {edition} {2nd}\ ed.\ (\bibinfo  {publisher}
  {Springer-Verlag},\ \bibinfo {year} {2009})\BibitemShut {NoStop}%
\bibitem [{\citenamefont {Candes}\ and\ \citenamefont
  {Recht}(2009)}]{candes09}%
  \BibitemOpen
  \bibfield  {author} {\bibinfo {author} {\bibfnamefont {E.~J.}\ \bibnamefont
  {Candes}}\ and\ \bibinfo {author} {\bibfnamefont {B.}~\bibnamefont {Recht}},\
  }\href@noop {} {\bibfield  {journal} {\bibinfo  {journal} {Found. Comput.
  Math.}\ }\textbf {\bibinfo {volume} {9}},\ \bibinfo {pages} {717} (\bibinfo
  {year} {2009})}\BibitemShut {NoStop}%
\bibitem [{\citenamefont {Sorensen}(1998)}]{sorensen98}%
  \BibitemOpen
  \bibfield  {author} {\bibinfo {author} {\bibfnamefont {D.}~\bibnamefont
  {Sorensen}},\ }\href@noop {} {\emph {\bibinfo {title} {Deflation for
  implicitly restarted {A}rnoldi methods}}},\ \bibinfo {type} {Tech. Rep.}\
  (\bibinfo  {institution} {CAAM at Rice University},\ \bibinfo {year}
  {1998})\BibitemShut {NoStop}%
\bibitem [{\citenamefont {Sorensen}(2002)}]{sorensen02}%
  \BibitemOpen
  \bibfield  {author} {\bibinfo {author} {\bibfnamefont {D.}~\bibnamefont
  {Sorensen}},\ }\href@noop {} {\bibfield  {journal} {\bibinfo  {journal} {Acta
  Numer.}\ ,\ \bibinfo {pages} {519}} (\bibinfo {year} {2002})}\BibitemShut
  {NoStop}%
\bibitem [{\citenamefont {Larsen}(1998)}]{larsen98}%
  \BibitemOpen
  \bibfield  {author} {\bibinfo {author} {\bibfnamefont {R.~M.}\ \bibnamefont
  {Larsen}},\ }\href@noop {} {\emph {\bibinfo {title} {Lanczos
  bidiagonalization with partial reorthogonalization}}},\ \bibinfo {type}
  {Tech. Rep.}\ \bibinfo {number} {DAIMI PB-357}\ (\bibinfo  {institution}
  {Department of Computer Science, Aarhus University},\ \bibinfo {year}
  {1998})\BibitemShut {NoStop}%
\bibitem [{\citenamefont {Saad}(1992)}]{saas92}%
  \BibitemOpen
  \bibfield  {author} {\bibinfo {author} {\bibfnamefont {Y.}~\bibnamefont
  {Saad}},\ }\href@noop {} {\bibfield  {journal} {\bibinfo  {journal} {SIAM J.
  Numer. Anal.}\ }\textbf {\bibinfo {volume} {29}},\ \bibinfo {pages} {209}
  (\bibinfo {year} {1992})}\BibitemShut {NoStop}%
\bibitem [{\citenamefont {Hochbruck}\ and\ \citenamefont
  {Lubich}(1997)}]{hochbruck97}%
  \BibitemOpen
  \bibfield  {author} {\bibinfo {author} {\bibfnamefont {M.}~\bibnamefont
  {Hochbruck}}\ and\ \bibinfo {author} {\bibfnamefont {C.}~\bibnamefont
  {Lubich}},\ }\href {\doibase 10.1137/S0036142995280572} {\bibfield  {journal}
  {\bibinfo  {journal} {SIAM J. Numer. Anal.}\ }\textbf {\bibinfo {volume}
  {34}},\ \bibinfo {pages} {1911} (\bibinfo {year} {1997})}\BibitemShut
  {NoStop}%
\bibitem [{\citenamefont {Nielsen}\ and\ \citenamefont
  {Chuang}(2011)}]{nielsen11}%
  \BibitemOpen
  \bibfield  {author} {\bibinfo {author} {\bibfnamefont {M.~A.}\ \bibnamefont
  {Nielsen}}\ and\ \bibinfo {author} {\bibfnamefont {I.~L.}\ \bibnamefont
  {Chuang}},\ }\href@noop {} {\emph {\bibinfo {title} {Quantum Computation and
  Quantum Information}}},\ \bibinfo {edition} {10th}\ ed.\ (\bibinfo
  {publisher} {Cambridge University Press},\ \bibinfo {address} {New York, NY,
  USA},\ \bibinfo {year} {2011})\BibitemShut {NoStop}%
\bibitem [{\citenamefont {Plenio}\ and\ \citenamefont
  {Virmani}(2007)}]{PlenioVirmani2007}%
  \BibitemOpen
  \bibfield  {author} {\bibinfo {author} {\bibfnamefont {M.~B.}\ \bibnamefont
  {Plenio}}\ and\ \bibinfo {author} {\bibfnamefont {S.}~\bibnamefont
  {Virmani}},\ }\href@noop {} {\bibfield  {journal} {\bibinfo  {journal}
  {Quant. Inf. Comp.}\ }\textbf {\bibinfo {volume} {7}},\ \bibinfo {pages} {1}
  (\bibinfo {year} {2007})}\BibitemShut {NoStop}%
\bibitem [{\citenamefont {Calabrese}\ and\ \citenamefont
  {Lefevre}(2008)}]{calabrese08}%
  \BibitemOpen
  \bibfield  {author} {\bibinfo {author} {\bibfnamefont {P.}~\bibnamefont
  {Calabrese}}\ and\ \bibinfo {author} {\bibfnamefont {A.}~\bibnamefont
  {Lefevre}},\ }\href@noop {} {\bibfield  {journal} {\bibinfo  {journal} {Phys.
  Rev. A}\ }\textbf {\bibinfo {volume} {78}},\ \bibinfo {pages} {032329}
  (\bibinfo {year} {2008})}\BibitemShut {NoStop}%
\bibitem [{\citenamefont {Schuch}\ \emph {et~al.}(2008)\citenamefont {Schuch},
  \citenamefont {Wolf}, \citenamefont {Verstraete},\ and\ \citenamefont
  {Cirac}}]{schuch08}%
  \BibitemOpen
  \bibfield  {author} {\bibinfo {author} {\bibfnamefont {N.}~\bibnamefont
  {Schuch}}, \bibinfo {author} {\bibfnamefont {M.~M.}\ \bibnamefont {Wolf}},
  \bibinfo {author} {\bibfnamefont {F.}~\bibnamefont {Verstraete}}, \ and\
  \bibinfo {author} {\bibfnamefont {J.~I.}\ \bibnamefont {Cirac}},\ }\href@noop
  {} {\bibfield  {journal} {\bibinfo  {journal} {Phys. Rev. Lett.}\ }\textbf
  {\bibinfo {volume} {100}},\ \bibinfo {pages} {030504} (\bibinfo {year}
  {2008})}\BibitemShut {NoStop}%
\bibitem [{bla()}]{blas}%
  \BibitemOpen
  \href {http://www.netlib.org/blas/} {\enquote {\bibinfo {title}
  {http://www.netlib.org/blas/},}\ }\BibitemShut {NoStop}%
\bibitem [{lap({\natexlab{a}})}]{lapack}%
  \BibitemOpen
  \href {http://www.netlib.org/lapack/} {\enquote {\bibinfo {title}
  {http://www.netlib.org/lapack/},}\ } ({\natexlab{a}})\BibitemShut {NoStop}%
\bibitem [{mkl()}]{mkl}%
  \BibitemOpen
  \href {https://software.intel.com/en-us/intel-mkl} {\enquote {\bibinfo
  {title} {https://software.intel.com/en-us/intel-mkl},}\ }\BibitemShut
  {NoStop}%
\bibitem [{cub()}]{cublas}%
  \BibitemOpen
  \href {https://developer.nvidia.com/cuBLAS} {\enquote {\bibinfo {title}
  {https://developer.nvidia.com/cublas},}\ }\BibitemShut {NoStop}%
\bibitem [{cul()}]{cula}%
  \BibitemOpen
  \href {http://www.culatools.com} {\enquote {\bibinfo {title}
  {http://www.culatools.com},}\ }\BibitemShut {NoStop}%
\bibitem [{cbl()}]{cblas}%
  \BibitemOpen
  \href {http://www.netlib.org/blas/cblas} {\enquote {\bibinfo {title}
  {http://www.netlib.org/blas/cblas},}\ }\BibitemShut {NoStop}%
\bibitem [{lap({\natexlab{b}})}]{lapacke}%
  \BibitemOpen
  \href {http://www.netlib.org/lapack/lapacke.html} {\enquote {\bibinfo {title}
  {http://www.netlib.org/lapack/lapacke.html},}\ } ({\natexlab{b}})\BibitemShut
  {NoStop}%
\bibitem [{RRS()}]{RRSVDgit}%
  \BibitemOpen
  \href {https://github.com/kindaguy/RRSVD.git} {\enquote {\bibinfo {title}
  {https://github.com/kindaguy/rrsvd.git},}\ }\BibitemShut {NoStop}%
\bibitem [{\citenamefont {Voronin}\ and\ \citenamefont
  {Martinsson}(2015)}]{martinsson15}%
  \BibitemOpen
  \bibfield  {author} {\bibinfo {author} {\bibfnamefont {S.}~\bibnamefont
  {Voronin}}\ and\ \bibinfo {author} {\bibfnamefont {P.}~\bibnamefont
  {Martinsson}},\ }\href@noop {} {\bibfield  {journal} {\bibinfo  {journal}
  {arXiv:1502.05366 [math.NA]}\ } (\bibinfo {year} {2015})}\BibitemShut
  {NoStop}%
\bibitem [{\citenamefont {Amdahl}(1967)}]{amdahl67}%
  \BibitemOpen
  \bibfield  {author} {\bibinfo {author} {\bibfnamefont {G.~M.}\ \bibnamefont
  {Amdahl}},\ }\href@noop {} {\bibfield  {journal} {\bibinfo  {journal} {AFIPS
  Conf. Proc.}\ }\textbf {\bibinfo {volume} {30}},\ \bibinfo {pages} {483}
  (\bibinfo {year} {1967})}\BibitemShut {NoStop}%
\bibitem [{\citenamefont {Tamascelli}\ \emph {et~al.}(2014)\citenamefont
  {Tamascelli}, \citenamefont {Dambrosio}, \citenamefont {Conte},\ and\
  \citenamefont {Ceotto}}]{tama14}%
  \BibitemOpen
  \bibfield  {author} {\bibinfo {author} {\bibfnamefont {D.}~\bibnamefont
  {Tamascelli}}, \bibinfo {author} {\bibfnamefont {F.~S.}\ \bibnamefont
  {Dambrosio}}, \bibinfo {author} {\bibfnamefont {R.}~\bibnamefont {Conte}}, \
  and\ \bibinfo {author} {\bibfnamefont {M.}~\bibnamefont {Ceotto}},\
  }\href@noop {} {\bibfield  {journal} {\bibinfo  {journal} {J. Chem. Phys.}\
  }\textbf {\bibinfo {volume} {140}},\ \bibinfo {pages} {174109} (\bibinfo
  {year} {2014})}\BibitemShut {NoStop}%
\end{thebibliography}
\end{document}